%% file: Main.tex
\documentclass[10pt,journal]{IEEEtran}
\IEEEoverridecommandlockouts
\usepackage{array}
\usepackage[caption=false,font=normalsize,labelfont=sf,textfont=sf]{subfig}
\usepackage{textcomp}
\usepackage{stfloats}
\usepackage{url}
\usepackage{verbatim}
\usepackage{graphicx}
\usepackage{cite}
\usepackage{bm}
\usepackage{bbm}
\usepackage{color,soul}
\usepackage{multirow}
\usepackage{booktabs}
\usepackage[table,xcdraw]{xcolor}
\usepackage{enumitem}
\usepackage{verbatim}
\usepackage{mathrsfs} 
\usepackage{xcolor}
\usepackage{setspace}
\usepackage{hyperref} 
\usepackage{algorithmicx}
\usepackage{float}
\usepackage{amsmath,amsthm,amsfonts,amssymb}
\theoremstyle{plain}
\newtheorem{thm}{Theorem}
\newtheorem{lem}[thm]{Lemma}

\newtheorem{prop}[thm]{Proposition}

\newtheorem{rem}{Remark}
\newtheorem{sty1}{Theorem}
\newtheorem{defi}[sty1]{Definition}

\usepackage[english]{babel}
\usepackage{algorithm}
\usepackage[noend]{algpseudocode}

\usepackage{svg}
\svgpath{{Figure/}}  

\allowdisplaybreaks[4]

\def\BibTeX{{\rm B\kern-.05em{\sc i\kern-.025em b}\kern-.08em
    T\kern-.1667em\lower.7ex\hbox{E}\kern-.125emX}}
\usepackage{balance}

\begin{document}
\title{Agile Affine Frequency Division Multiplexing}

\author{Yewen Cao and Yulin Shao
\thanks{The authors are with the Department of Electrical and Electronic Engineering, The University of Hong Kong, Hong Kong S.A.R..}
\thanks{Correspondence: \url{ylshao@hku.hk}.}
}


\maketitle

\begin{abstract}
The advancement to 6G calls for waveforms that transcend static robustness to achieve intelligent adaptability. Affine Frequency Division Multiplexing (AFDM), despite its strength in doubly-dispersive channels, has been confined by chirp parameters optimized for worst-case scenarios. This paper shatters this limitation with Agile-AFDM, a novel framework that endows AFDM with dynamic, data-aware intelligence. By redefining chirp parameters as optimizable variables for each transmission block based on real-time channel and data information, Agile-AFDM transforms into an adaptive platform. It can actively reconfigure its waveform to minimize peak-to-average power ratio (PAPR) for power efficiency, suppress inter-carrier interference (ICI) for communication reliability, or reduce Cram\'er-Rao bound (CRLB) for sensing accuracy. This paradigm shift from a static, one-size-fits-all waveform to a context-aware signal designer is made practical by efficient, tailored optimization algorithms. Comprehensive simulations demonstrate that this capability delivers significant performance gains across all metrics, surpassing conventional OFDM and static AFDM. Agile-AFDM, therefore, offers a crucial step forward in the design of agile waveforms for 6G and beyond.
\end{abstract}


\section{Introduction}\label{Sec:Introduction}
\input{Introduction.tex}

\section{System Model}\label{Sec:System Model}
\input{System_Model.tex}

\section{Agile Waveforming for Power Efficiency}\label{Sec:PAPR}
\input{PAPR.tex}

\section{Agile Waveforming for Reliable Communications}\label{Sec:SIR}
\input{SIR.tex}

\section{Agile Waveforming for Sensing Accuracy}\label{Sec:CRLB}
\input{CRLB.tex}

\section{Numerical and Simulation Results}\label{Sec:Simu}
\input{Simu.tex}

\section{Conclusions}\label{Sec:Conclusion}
\input{Conclusions.tex}

\appendices
\section{Proof of Proposition \ref{Prop:perio}}\label{sec:AppA}
\input{Appendix/App_prop_perio.tex}

\section{Proof of Theorem \ref{thm:AFDM_trigonometric}}\label{sec:AppB}
\input{Appendix/App_thm_tri}

\bibliographystyle{IEEEtran}
\bibliography{References}



\end{document}

%% file: Introduction.tex
\subsection{Background}
The evolution of wireless communication has been fundamentally shaped by its core waveform, driven relentlessly by demands for higher data rates and lower latency \cite{lien20175g,dang2020should,liu2022integrated,shao2024theory}. In 4G and 5G systems, Orthogonal Frequency Division Multiplexing (OFDM) has served this role effectively \cite{lien20175g,cho2010mimo}, prized for its robustness against frequency-selective fading and implementation efficiency. However, the emerging vision for 6G, encompassing high-speed mobility \cite{6Gvision}, unmanned aerial vehicles \cite{UAV}, and Integrated Sensing and Communication (ISAC) \cite{ISACsurvey}, exposes a critical vulnerability in OFDM: acute susceptibility to Doppler spread. In high-mobility environments, Doppler effects disrupt subcarrier orthogonality, leading to severe Inter-Carrier Interference (ICI) and significant performance degradation \cite{cho2010mimo,shao2021federated}.

This limitation has catalyzed the search for waveforms inherently robust to time-varying channels \cite{hlawatsch2011wireless,OTFS,bemani_affine_2023_twc}. Among promising candidates, Affine Frequency Division Multiplexing (AFDM) \cite{bemani_affine_2023_twc} has emerged as a powerful and generalized multi-carrier framework. Its core innovation lies in the Discrete Affine Fourier Transform (DAFT), which generalizes the classical discrete Fourier transform (DFT) by incorporating two chirp parameters, $c_1$ and $c_2$. These parameters provide unparalleled flexibility: $c_1$ dictates the linear chirp rate in the time domain, acting as a pre-distortion to combat multipath delay spread, while $c_2$ applies a quadratic phase pre-coding in the Discrete Affine Fourier (DAF) domain, ensuring linear independence among path responses \cite{rou2025affine}.
By pre-setting these parameters according to the statistical worst-case channel bounds (e.g., maximum delay and Doppler), AFDM can orthogonalize doubly-dispersive channels \cite{bemani_affine_2023_twc}, achieving full diversity order and superior resilience in high-mobility scenarios.

\subsection{Chirp Parameter Configuration}
The chirp parameters $c_1$ and $c_2$ are the defining degrees of freedom in AFDM, sculpting its time-frequency structure to combat double dispersion. Existing configuration strategies can be categorized by their temporal adaptability and optimization objectives.

The foundational and most prevalent approach employs \textit{static} parameters, determined during system design and fixed throughout operation. This paradigm is inherently conservative: parameters are calculated based on statistical worst-case channel characteristics, such as the maximum expected delay and Doppler spreads, to guarantee baseline performance under the most adverse conditions \cite{bemani_affine_2023_twc,zhou_affine_2025_mvt,bemani_integrated_2024_lwc,rou_from_2024_msp,ranasinghe_joint_2025_twc}. The primary objective is to ensure fundamental properties like path separability and the achievement of full diversity order. Following this principle, $c_1$ is typically determined by the maximum tolerable Doppler shift, while $c_2$ is set as an irrational number or a sufficiently small rational to ensure linear independence among channel paths in the DAF domain.

Within this static framework, subsequent research has sought to optimize the fixed $(c_1, c_2)$ pair for more specific objectives. For instance, \cite{li_chirp_2025_tcom} tunes $c_1$ to minimize the bit error rate (BER) under minimum mean square error (MMSE) equalization by strategically separating dominant off-diagonal components of the effective channel matrix; \cite{benzine_affine_2024_spawc} adjusts $c_1$ based on the hierarchical sparsity levels of the delay-Doppler profile to satisfy the Restricted Isometry Property (RIP) for compressed sensing.

In the context of ISAC, \cite{yin_ambiguity_2025_jsac} demonstrates that $c_1$ shapes the unambiguity parallelogram of the ambiguity function, which has to match sensing channel characteristics for interference-free parameter estimation. Ref. \cite{ni_an_2025_twc} establishes analytical relationships between $c_1$ and sensing metrics, enabling a flexible trade-off between the maximum tolerable delay and Doppler shifts.
While these works optimize for a particular metric, the resulting parameters remain fixed and cannot adapt to the instantaneous channel state or the data being transmitted.

A separate line of research introduces variability in $(c_1, c_2)$, but primarily to enable auxiliary functionalities rather than to dynamically optimize core communication/sensing performance. This includes using parameter variation as an additional resource for data transmission or system security. For example, Index Modulation (IM) schemes exploits $c_2$ as a data-carrying resource \cite{liu_pre_2025_twc}. By selecting $c_2$ from a predefined alphabet, pre-chirp-domain index modulation enhances spectral efficiency through implicit signaling while maintaining full diversity in doubly dispersive channels.
Further, \cite{chen_chirp_2025_icc} exploits dynamic chirp parameter configuration as an encryption mechanism for physical layer security. By implementing a key-driven pseudo-random hopping pattern for $(c_1, c_2)$, the signal constellation is randomized for eavesdroppers, thereby preventing information leakage without complex upper-layer encryption. A particularly relevant and parallel development is the dynamic adjustment of $c_2$ \cite{cao2025fractional,yuan_papr_2025_lwc} to combat high Peak-to-Average Power Ratio (PAPR). In \cite{yuan_papr_2025_lwc}, the authors propose a grouped pre-chirp selection (GPS) approach to explore a discrete set of candidate $c_2$ values to identify the one yielding the lowest instantaneous PAPR for transmission.

\subsection{Contributions}
The prevailing paradigm for AFDM, i.e., optimizing chirp parameters offline for worst-case channels, is fundamentally static, failing to exploit the waveform's inherent agility. This underutilization limits AFDM's potential as a truly adaptive system for intelligent and efficient 6G.

To bridge this gap and fully unlock the potential of AFDM, we introduce Agile-AFDM, a novel framework that transforms AFDM from a static waveform into a dynamic, intelligent, and data-aware platform. Our core innovation can be succinctly captured as follows:

\textit{Agile-AFDM is a waveform adaptation technique that performs dynamic, per-block optimization of the AFDM chirp parameters, driven by real-time Channel State Information (CSI), the specific data symbols being transmitted, and targeted performance objectives. It serves as a universal plug-and-play module, empowering AFDM systems to excel across diverse metrics such as PAPR (for enhanced power efficiency), ICI (for improved communication reliability), and Cram\'er-Rao Lower Bound (CRLB, for superior sensing accuracy).}

This adaptive philosophy finds a direct parallel in the evolution of the Cyclic Prefix (CP) in OFDM systems: from a static length designed for the worst-case delay spread (as in 4G LTE \cite{dahlman20134g}), to flexible configurations that adapt to specific scenarios and efficiency demands (as in 5G New Radio \cite{lien20175g}). Similarly, Agile-AFDM shifts the chirp parameters from static, worst-case guarantees to dynamic, context-aware optima, unlocking significant gains in efficiency, reliability, and sensing precision for each transmission block.

This work makes three key contributions that translate this vision into a concrete and effective framework:
\begin{itemize}[leftmargin=0.5cm]
    \item We put forth a paradigm shift in AFDM design by formulating chirp parameter selection as a per-data-block dynamic optimization problem. Rather than relying on universal static values, our approach determines the optimal parameters in real time based on the instantaneous channel state, the actual data symbols, and the targeted performance objective. This transforms AFDM into an agile, reconfigurable waveform platform, enabling on-demand customization for communication and sensing, and fundamentally overcoming the limitations of static configurations.
    \item We provide a comprehensive analytical and algorithmic foundation for agile waveform adaptation. This includes a rigorous analysis of how the chirp parameters $(c_1,c_2)$ govern three core metrics: PAPR, ICI, and CRLB. We reveal that PAPR is primarily governed by $c_2$ and exhibits a periodic structure, whereas ICI and CRLB are jointly determined by both parameters. Building on these insights, we introduce: (i) a low-complexity fine-grained search algorithm for PAPR reduction that exploits parameter periodicity; (ii) a Fractional Programming (FP)-based alternating optimization method for ICI suppression via Signal-to-Interference Ratio (SIR) maximization; and (iii) a Particle Swarm Optimization (PSO) scheme to navigate the non-convex landscape of CRLB minimization for sensing parameter estimation.
    \item Through extensive numerical simulations, we demonstrate that Agile-AFDM delivers significant performance improvements over static configurations across diverse channel conditions.
    Compared to OFDM and static AFDM, Agile-AFDM reduces PAPR by approximately 50\%, outperforming other state-of-the-art techniques. For communication reliability, Agile-AFDM achieves a $7.32$ dB SIR gain over OFDM and a $2.64$ dB gain over static AFDM, with lower variance indicating greater robustness. In sensing, Agile-AFDM lowers the CRLB by an average of $8.3$ dB for delay estimation and $12.7$ dB for Doppler estimation, with particularly pronounced improvements in high-mobility scenarios. These results validate the framework's effectiveness and its transformative potential for next-generation adaptive wireless systems.
\end{itemize}

{\it Notations}: We use $\Re$ and $\Im$ to denote the real and imaginary parts of a complex number, respectively.
The imaginary unit is represented by $j$. 
For a vector or matrix, $(\cdot)^*$ denotes the complex conjugate, and $(\cdot)^H$ denotes the conjugate transpose (Hermitian transpose).
$\mathbb{Z}$ represents the set of integers. 
$\mathbb{R}$ represents the set of real numbers.
$\mathcal{C}$ and $\mathcal{CN}$ stand for the real and complex Gaussian distributions, respectively.
Additionally, $\mathcal{U}(a, b)$ denotes the uniform distribution over the interval $[a, b]$.
$\bm{F}^{-1}$ denotes the inverse of matrix $\bm{F}$. $\mathbb{E}_m[\cdot]$ denotes the expectation with respect to $m$. 
$\oslash$ is the element-wise division (also known as Hadamard division), denoting the operation of dividing corresponding elements of two same-dimensional matrices or vectors pairwise.
$\nabla_{\bm{c}}$ denotes the gradient with respect to vector $\bm{c}$.
The inner product of two vector $\bm{a}$ and $\bm{b}$ is denoted by $\langle \bm{a}, \bm{b} \rangle = \bm{a}^H \bm{b}$.

%% file: System_Model.tex
\begin{figure*}[t]
  \centering
  \includegraphics[width=0.8\textwidth]{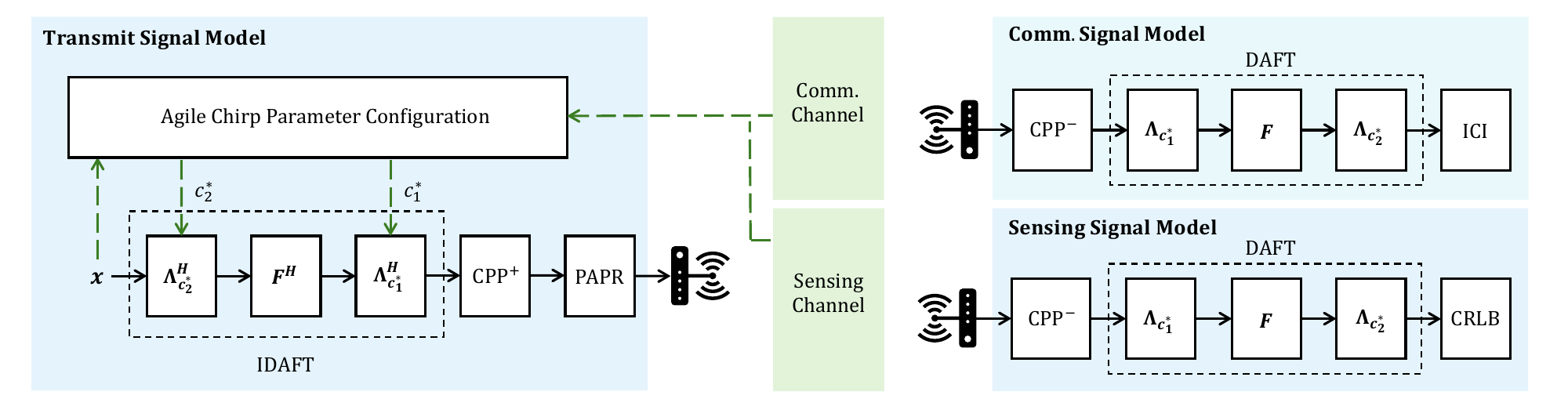}
  \caption{Architecture of the Agile-AFDM transceiver, showcasing the signal processing flow and its adaptive configuration for three objectives: minimizing PAPR, mitigating ICI, and lowering CRLB for sensing.}
  \label{Fig:framework}
\end{figure*}

The core tenet of Agile-AFDM is to transcend the static configuration paradigm by dynamically optimizing the chirp parameters, $c_1$ and $c_2$, for each transmission block. This optimization is conditioned upon the instantaneous CSI, the data symbols to be transmitted $\bm{x}$, and a specific, real-time performance objective. To demonstrate its versatility, this work focuses on three critical objectives: minimizing PAPR for power-efficient transmission, mitigating ICI for reliable communications, and lowering the CRLB for accurate sensing. The following subsections present the transceiver structure of Agile-AFDM, as illustrated in Fig. \ref{Fig:framework}, and derive these three metrics, highlighting the pivotal role of the chirp parameters in shaping each one.

\subsection{Transmit Signal Model and PAPR}
We begin at the transmitter, where the interplay between the data symbols and the chirp parameters dictates the waveform's structure. Let $\bm{x} = [x[0], x[1], \cdots, x[N-1]]^\top$ denote a block of data symbols in the DAF domain, where each element $x[m]$ is drawn from a Quadrature Amplitude Modulation (QAM) constellation or a complex Gaussian distribution.

The transmitter maps $\bm{x}$ to the time-domain signal $\bm{s}$ via the inverse DAFT (IDAFT), parameterized by $(c_1, c_2)$:
\begin{equation}
\bm{s} = \bm{A}^H \bm{x} = \bm{\Lambda}_{c_1}^H \bm{F}^H \bm{\Lambda}_{c_2}^H \bm{x},
\end{equation}
where $\bm{F}$ is the DFT matrix, and $\bm{\Lambda}_{c_1}$, $\bm{\Lambda}_{c_2}$ are diagonal matrices with elements $[\bm{\Lambda}_{c_1}]_{n,n} = e^{-j 2\pi c_1 n^2}$ and $[\bm{\Lambda}_{c_2}]_{m,m} = e^{-j 2\pi c_2 m^2}$, respectively. 

The element-wise representation of the time-domain modulated signal is given by
\begin{equation} \label{Eq:s_n}
s[n] = \sum_{m=0}^{N-1} x[m] \phi_n(m), \quad \text{for } n = 0, 1, \ldots, N-1,
\end{equation}
where the basis function is defined as
\begin{equation}
\phi_n(m) \triangleq \frac{1}{\sqrt{N}} e^{j 2\pi \left(c_1 n^2 + c_2 m^2 + \frac{n m}{N}\right)}.
\end{equation}

To combat multipath propagation, a chirp-periodic prefix (CPP) is appended, and we define
\begin{equation*}
s[n] \triangleq s[N+n] e^{-j 2\pi c_1 (N^2 + 2Nn)},~\text{for } n = -L_{cp}, \dots, -1.
\end{equation*}
Note that when $2Nc_1$ is an integer and $N$ is even, CPP becomes equivalent to the standard CP used in OFDM.

A critical performance metric, especially for uplink transmission where user equipment is power-limited, is the PAPR. It quantifies the envelope fluctuations of the transmitted signal and is defined for the continuous-time baseband signal $s(t)$, as the power amplifier operates in the analog domain and the peak power may occur between sampling instants. The continuous signal $s(t)$ can be viewed as an interpolated version of the discrete samples, where the discrete sample $s[n]$ corresponds to the value of $s(t)$ at the sampling instant $t = n T_s$, i.e., $s[n] = s(n T_s)$.

\begin{defi}
The PAPR of the AFDM system is defined as
\begin{equation}\label{Eq:PAPR_def}
    \xi \triangleq \frac{\max_{0 \le t \le T_{\text{sym}}}\big\{|s(t)|^2\big\}}{\mathbb{E}\big[|s(t)|^2\big]},~~~
    \xi_{\text{dB}} \triangleq 10\log\xi,
\end{equation}
where $T_{\text{sym}} \triangleq N T_{s}$ is the block duration, and the continuous-time baseband signal $s(t)$ is given by
\begin{equation}\label{Eq:s(t)}
\begin{aligned}
s(t) = \sqrt{\frac{1}{N}}\sum\limits_{m=0}^{N-1} x[m] e^{j \frac{2 \pi}{T_s^2} c_{1} t^2}  e^{j 2 \pi c_{2} m^2} e^{j2\pi m t/T_{\text{sym}}}.
\end{aligned}
\end{equation}
\end{defi}

\begin{rem}
   Equation \eqref{Eq:s(t)} provides the analytical expression for $s(t)$, which maps the discrete symbols $x[m]$ to a continuous-time waveform through the inverse unitary transform. By letting $t = n T_s$ in \eqref{Eq:s(t)}, the discrete version of the AFT (as in \eqref{Eq:s_n}) can be directly derived, thereby establishing the consistency between discrete and continuous signals. 
\end{rem}

In Section \ref{Sec:PAPR}, we will analyze the relationship between the chirp parameters and the PAPR, revealing that it is primarily governed by $c_2$ and exhibits a periodic structure. This enables efficient, per-block optimization of $c_2$ to minimize PAPR, achieving adaptive waveform tailoring for power efficiency.

\subsection{Communication Signal Model and ICI}
We now turn to the communication receiver to examine how the chirp parameters affect communication reliability in doubly-dispersive channels. Consider the transmitted signal $\bm{s}$ propagating through a time-varying multipath channel to a receiver. The channel impulse response is
\begin{equation*}
g^{\text{comm}}_n(\ell) = \sum_{i=1}^{P} h_i e^{-j 2\pi f_i n} \delta(\ell - \ell_i),
\end{equation*}
where $P$ is the number of paths, each characterized by a complex gain $h_i$, a Doppler shift $f_i$ (normalized digital frequency), and a delay $\ell_i$ (in samples).

After CPP removal, the received signal can be written as
\begin{equation*}
\bm{r}^{\text{comm}} = \bm{H}^{\text{comm}} \bm{s} + \bm{w}^{\text{comm}},
\end{equation*}
where $\bm{w}^{\text{comm}} \sim \mathcal{CN}(\bm{0}, N_0 \bm{I})$ is complex additive white Gaussian noise (AWGN), and $\bm{H}^{\text{comm}}$ is the channel convolution matrix. This matrix can be decomposed as
\begin{equation}\label{eq:commH}
\bm{H}^{\text{comm}} = \sum_{i=1}^{P} h_i \boldsymbol{\Gamma}_{\text{CPP},i} \boldsymbol{\Delta}_{f_i} \boldsymbol{\Pi}^{l_i},
\end{equation}
where $\boldsymbol{\Delta}_{f} \triangleq \operatorname{diag}([e^{-j 2\pi f n}]_{n=0}^{N-1})$ is the Doppler shift matrix; $\boldsymbol{\Pi}$ is the forward cyclic-shift matrix with $[\boldsymbol{\Pi}]_{i,j} \triangleq \delta((i-j+1)\mod N)$; $\boldsymbol{\Gamma}_{\text{CPP},i}$ is the CPP compensation matrix for the $i$-th path.

Applying the DAFT at the receiver converts the signal back to the DAF domain, yielding
\begin{equation}\label{eq:commy}
\bm{y}^{\text{comm}} = \bm{A} \bm{r}^{\text{comm}} = \bm{H}_{\text{eff}}^{\text{comm}} \bm{x} + \tilde{\bm{w}}^{\text{comm}},
\end{equation}
where $\tilde{\bm{w}}^{\text{comm}} = \bm{A} \bm{w}^{\text{comm}} \sim \mathcal{CN}(\bm{0}, N_0\bm{I})$ is the DAF-domain noise, and $\bm{H}_{\text{eff}}^{\text{comm}} \triangleq \bm{A} \bm{H}^{\text{comm}} \bm{A}^H$ is the effective end-to-end channel matrix in the DAF domain. Its elements are given by
\begin{equation}\label{eq:commHelement}
H^{\text{comm}}_{\text{eff}}[p,q] = \frac{1}{N} \sum_{i=1}^{P} h_i e^{j \frac{2\pi}{N} \left(N c_1 \ell_i^2 - q \ell_i + N c_2 (q^2 - p^2)\right)} \mathcal{F}_i(p,q),
\end{equation}
wherein we define
\begin{equation}\label{eq:commF}
\mathcal{F}_i(p, q) \triangleq \sum_{n=0}^{N-1} e^{-j \frac{2\pi}{N} \psi_i n} =
\begin{cases} 
N, & \text{if } \dfrac{\psi_i}{N} \in \mathbb{Z}, \\[10pt]
\dfrac{ e^{-j 2\pi \psi_i} - 1 }{ e^{-j \frac{2\pi}{N} \psi_i} - 1 }, & \text{otherwise,}
\end{cases}
\end{equation}
with $\psi_i \triangleq p - q + \nu_i + 2 N c_1 \ell_i$ and $\nu_i \triangleq N f_i$ being the normalized Doppler shift. The condition $\dfrac{\psi_i}{N} \in \mathbb{Z}$ corresponds to an alignment where the phase rotation per sample is a multiple of $2\pi$, leading to coherent superposition and a peak value of $N$. Otherwise, the function exhibits a sinc-like behavior, causing the interference to be distributed.

This structure directly determines the ICI. The input-output relationship in \eqref{eq:commy} can be expanded for the $p$-th received subcarrier as
\begin{equation}\label{eq:commy_subcarrier}
y^{\text{comm}}[p] = \underbrace{H_{\text{eff}}^{\text{comm}}[p,p] x[p]}_{\text{Desired signal}} + \underbrace{\sum_{\substack{q=0 \\ q \neq p}}^{N-1} H_{\text{eff}}^{\text{comm}}[p,q] x[q]}_{\text{ICI}} + \tilde{w}^{\text{comm}}[p]
\end{equation}

In high-mobility scenarios with large Doppler spreads, the ICI term becomes significant and severely degrades the decoding performance. The core advantage of Agile-AFDM lies in its ability to mitigate this issue. By dynamically optimizing the chirp parameters $c_1$ and $c_2$ based on the instantaneous CSI and the data vector $\bm{x}$, we can reshape the effective channel matrix $\bm{H}_{\text{eff}}^{\text{comm}}$. The goal is to suppress the off-diagonal elements that constitute the ICI, thereby enhancing the signal-to-interference ratio (SIR). A detailed analysis of the SIR and the corresponding optimization algorithm is presented in Section \ref{Sec:SIR}.

\subsection{Sensing Signal Model and CRLB}
Beyond reliable communications, the convergence of wireless systems now demands high-accuracy sensing capabilities, a cornerstone of ISAC. Agile-AFDM natively supports this dual functionality. We now consider a monostatic sensing scenario, where the transmitter illuminates a target and processes the echo of its own transmitted signal $\bm{s}$ to estimate target parameters.

To maintain structural consistency with the communication model, we formulate the sensing model in an analogous manner. The sensing channel impulse response from the transmitter to a target and back is characterized by
\begin{equation}
g^{\text{sens}}_{n}(\ell) = \alpha e^{-j 2\pi f n} \delta(\ell - \ell),
\end{equation}
where $\alpha$ is the complex reflection coefficient, $f$ is the Doppler shift (normalized digital frequency), and $l$ is the two-way delay (in samples).

The received echo at the transmitter, after CPP removal, is given by
\begin{equation}
\bm{r}^{\text{sens}} = \bm{H}^{\text{sens}} \bm{s} + \bm{w}^{\text{sens}},
\end{equation}
where $\bm{w}^{\text{sens}} \sim \mathcal{CN}(\bm{0}, N_0 \bm{I})$ is AWGN, and the sensing channel matrix mirrors its communication counterpart:
\begin{equation*}
\bm{H}^{\text{sens}} = \alpha \boldsymbol{\Gamma}_{\text{CPP}} \boldsymbol{\Delta}_{f} \boldsymbol{\Pi}^{l}.
\end{equation*}

Transforming the echo signal to the DAF domain yields the input-output relationship
\begin{equation}\label{eq:sensingy}
\bm{y}^{\text{sens}} = \bm{A} \bm{r}^{\text{sens}} = \bm{H}_{\text{eff}}^{\text{sens}} \bm{x} + \tilde{\bm{w}}^{\text{sens}},
\end{equation}
where $\bm{H}_{\text{eff}}^{\text{sens}} \triangleq \bm{A} \bm{H}^{\text{sens}} \bm{A}^H$ is the effective sensing channel matrix, and $\tilde{\bm{w}}^{\text{sens}} = \bm{A} \bm{w}^{\text{sens}} \sim \mathcal{CN}(\bm{0}, N_0\bm{I})$.
Similar to \eqref{eq:commHelement}, the elements of $\bm{H}_{\text{eff}}^{\text{sens}}$ can be written as
\begin{equation*}
H^{\text{sens}}_{\text{eff}}[p,q] = \frac{1}{N} \alpha e^{j \frac{2\pi}{N} \left(N c_1 \ell^2 - q \ell + N c_2 (q^2 - p^2)\right)} \mathcal{F}(p,q),
\end{equation*}
where $\mathcal{F}$ is as defined in \eqref{eq:commF}.

The primary goal in sensing is to estimate the target's delay $\ell$ and Doppler shift $\nu=Nf$ from the received signal $\bm{y}^{\text{sens}}$.
The CRLB provides a fundamental lower bound on the variance of any unbiased estimator of these parameters. 
As will be derived in Section \ref{Sec:CRLB}, the CRLB for both delay and Doppler estimation are intricate functions of the chirp parameters $c_1$ and $c_2$.
This inherent dependency enables Agile-AFDM to dynamically optimize these parameters, thereby minimizing the CRLB and achieving superior sensing accuracy compared to OFDM and static AFDM.

A key property that facilitates this parameter optimization is the periodicity of the effective channel matrices, which we formalize as follows.

\begin{prop}\label{Prop:perio}
The effective communication channel matrix $\bm{H}_{\text{eff}}^{\text{comm}}$ and the effective sensing channel matrix $\bm{H}_{\text{eff}}^{\text{sens}}$ are periodic functions of the chirp parameters $c_1$ and $c_2$ with period $1$. Specifically, for any integers $k, m \in \mathbb{Z}$, the following equalities hold:
\begin{align}
\bm{H}_{\text{eff}}^{\text{comm}}(c_1 + k, c_2 + m) &= \bm{H}_{\text{eff}}^{\text{comm}}(c_1, c_2), \\
\bm{H}_{\text{eff}}^{\text{sens}}(c_1 + k, c_2 + m) &= \bm{H}_{\text{eff}}^{\text{sens}}(c_1, c_2).
\end{align}
\end{prop}

\begin{proof}
    See Appendix \ref{sec:AppA}.
\end{proof}

This periodicity property confines the optimal parameter search to the compact space $(c_1,c_2)\in[0,1)\times[0,1)$, significantly reducing the computational complexity of the dynamic optimization process.

With the signal model established and the three core metrics defined, we now proceed to the analysis of the Agile-AFDM framework. The following three sections will systematically analyze the unique characteristics of each metric and present our specialized approaches for dynamic chirp parameter adaptation.

%% file: PAPR.tex

A critical drawback of conventional multi-carrier waveforms is their high PAPR \cite{jiang2005exponential}, which severely degrades power efficiency by forcing power amplifiers to operate in inefficient linear regions to avoid signal distortion. This challenge is particularly acute in uplink scenarios, where user equipment faces stringent power constraints. In this section, we demonstrate that Agile-AFDM addresses this fundamental limitation through adaptive chirp parameter optimization, offering a unique and effective PAPR reduction capability that transcends conventional waveform designs.

\subsection{PAPR Analysis}
The PAPR of Agile-AFDM, denoted as $\xi$, is defined in \eqref{Eq:PAPR_def}. To characterize $\xi$, we first analyze the envelope power function of $s(t)$, i.e., $|s(t)|^2$. Through detailed expansion of the continuous-time signal representation in \eqref{Eq:s(t)}, we obtain
\begin{eqnarray}\label{eq:AFDM_PAPR_expansion}
&&\hspace{-0.4cm} |s(t)|^2 = s(t)\cdot (s(t))^* = \frac{1}{N} \sum_{m=0}^{N-1} x[m]x^*[m] + \notag\\
&&\hspace{-0.4cm} \frac{2}{N} \left\{ \sum_{p=1}^{N-1} \gamma_{p}^{(1)}(c_{2})\cos(\frac{2\pi pt}{T}) + \sum_{p=1}^{N-1} \gamma_{p}^{(2)}(c_{2})\cos(\frac{2\pi pt}{T}) \right. \notag\\
&&\hspace{-0.4cm} \quad \left. + \sum_{p=1}^{N-1} \gamma_{p}^{(3)}(c_{2})\sin(\frac{2\pi pt}{T}) + \sum_{p=1}^{N-1} \gamma_{p}^{(4)}(c_{2})\sin(\frac{2\pi pt}{T}) \right\} \notag\\
&&\hspace{-0.4cm} \triangleq \frac{1}{N} \sum_{m=0}^{N-1} x[m]x^*[m] + \frac{2}{N} \{g(t)\}, 
\end{eqnarray}
wherein we have defined the constituent components as
\begin{equation*}
\gamma_{p}^{(1)}(c_{2}) \triangleq 
\sum\limits_{m=0}^{N-1-p}\lambda_{m,p}\cos\beta_{m,p}(c_{2}),
\end{equation*}
\begin{equation*}
\gamma_{p}^{(2)}(c_{2})
\triangleq -\sum\limits_{m=0}^{N-1-p}\mu_{m,p}\sin\beta_{m,p}(c_{2}),
\end{equation*}
\begin{equation*}
\gamma_{p}^{(3)}(c_{2})\triangleq -\sum\limits_{m=0}^{N-1-p}\lambda_{m,p}\sin\beta_{m,p}(c_{2}),
\end{equation*}
\begin{equation*}
\gamma_{p}^{(4)}(c_{2}) \triangleq -\sum\limits_{m=0}^{N-1-p}\mu_{m,p}\cos\beta_{m,p}(c_{2}),
\end{equation*}
with the fundamental phase term and data-dependent coefficients given by
\begin{equation*}
\beta_{m,p}(c_{2}) \triangleq 2\pi c_{2} p(2m+p),
\end{equation*}
\begin{equation*}
\lambda_{m,p} \triangleq \Re\{x[m+p]x^*[m]\},~ \mu_{m,p} \triangleq  \Im\{x[m+p]x^{*}[m]\}.
\end{equation*}

A crucial insight from \eqref{eq:AFDM_PAPR_expansion} is the independence of PAPR from parameter $c_{1}$. The chirp parameter $c_{2}$ exclusively governs PAPR through the phase terms $\beta_{m,p}(c_{2}) = 2\pi c_{2} p(2m+p)$, enabling us to express PAPR as a univariate function $\xi(c_{2})$. This parametric simplification is further enhanced by the periodicity of PAPR with respect to $c_{2}$, which significantly streamlines the optimization framework. Building upon this analytical foundation, we now establish two fundamental properties of the Agile-AFDM PAPR characteristic.

\begin{lem}\label{lem:AFDM_PAPR}
The PAPR of the Agile-AFDM system can be expressed as
\begin{equation*}\label{eq:AFDM_PAPR_formula}
\begin{aligned}
\xi_{\text{dB}} = 10\log\left(1 + \frac{2\max\{g(t)\}}{\sum\limits_{m=0}^{N-1} x[m]x^*[m]}\right),
\end{aligned}
\end{equation*}
where $g(t)$ is defined in \eqref{eq:AFDM_PAPR_expansion}.
\end{lem}

\begin{proof}
The expectation of $g(t)$ in \eqref{eq:AFDM_PAPR_expansion} is zero because it consists of a sum of trigonometric functions over a period $T$. Therefore, the average power of the signal is given by:
\begin{equation}\label{eq:AFDM_mean_power}
\mathbb{E}\big[|s(t)|^2\big] = \frac{1}{N}\sum_{m=0}^{N-1} x[m]x^*[m].
\end{equation}
Substituting \eqref{eq:AFDM_mean_power} into \eqref{Eq:PAPR_def} yields \eqref{eq:AFDM_mean_power}.
\end{proof}

\begin{lem}\label{thm:AFDM_PAPR_periodicity}
The PAPR of the Agile-AFDM system is a periodic function of $c_{2}$ with a period of $\frac{1}{2}$. This periodicity implies that within the range $c_{2} \in [0, \frac{1}{2})$, all possible PAPR values can be explored for a given data block $\mathbf{x}$.
\end{lem}

\begin{proof}
Lemma~\ref{lem:AFDM_PAPR} indicates that analyzing the periodicity of $\max\{g(t)\}$ is sufficient for determining the periodic behavior of PAPR.

To demonstrate the periodicity, define $\widetilde{c}_{2} \triangleq c_{2} + \frac{1}{2}$ and substitute it into the phase terms $\beta_{m,p}(c_{2}) = 2\pi c_{2} p(2m+p)$. This substitution yields:
\begin{equation*}
\beta_{m,p}(\widetilde{c}_{2}) = \beta_{m,p}(c_{2}) + \pi p(2m+p).
\end{equation*}

Since $p(2m+p)$ is an integer, we have:
\begin{equation*}
\cos\beta_{m,p}(\widetilde{c}_{2}) = (-1)^{p(2m+p)}\cos\beta_{m,p}(c_{2}),
\end{equation*}
\begin{equation*}
\sin\beta_{m,p}(\widetilde{c}_{2}) = (-1)^{p(2m+p)}\sin\beta_{m,p}(c_{2}).
\end{equation*}

The key observation is that the factor $(-1)^{p(2m+p)}$ can be absorbed by a time shift. Specifically, it can be shown that:
\begin{equation*}
(-1)^{p(2m+p)}\cos\left(\frac{2\pi pt}{T_{\text{sym}}}\right) = \cos\left(\frac{2\pi p(t+T_{\text{sym}}/2)}{T}\right),
\end{equation*}
\begin{equation*}
(-1)^{p(2m+p)}\sin\left(\frac{2\pi pt}{T_{\text{sym}}}\right) = \sin\left(\frac{2\pi p(t+T_{\text{sym}}/2)}{T}\right).
\end{equation*}

This time shift does not affect the maximum value of $g(t)$, as the trigonometric functions are periodic with period $T_{\text{sym}}$. Therefore, $\max\{g(t)\}$ remains unchanged when $c_{2}$ is increased by $\frac{1}{2}$, proving the periodicity of PAPR with respect to $c_{2}$ with period $\frac{1}{2}$.

Let $\widetilde{g}(t)$ represent the value of $g(t)$ when $c_{2} = \widetilde{c}_{2}$. Using the above relations, we can express $\widetilde{g}(t)$ as:
\begin{eqnarray*}
\widetilde{g}(t) &=& \sum_{i=1,2}\sum_{p=1}^{N-1} \gamma_p^{(i)}(c_{2}) \cos\left(\frac{2\pi p(t+T_{\text{sym}}/2)}{T_{\text{sym}}}\right) \\
&& + \sum_{j=3,4}\sum_{p=1}^{N-1} \gamma_p^{(j)}(c_{2}) \sin\left(\frac{2\pi p(t+T_{\text{sym}}/2)}{T_{\text{sym}}}\right).
\end{eqnarray*}

Since the trigonometric functions are periodic with period $T_{\text{sym}}$, shifting the argument by $T_{\text{sym}}/2$ does not alter the maximum value over one period. Therefore:
\begin{equation*}
\max\{\widetilde{g}(t)\} = \max\{g(t)\}.
\end{equation*}

This confirms that PAPR is periodic in $c_{2}$ with period $\frac{1}{2}$. Consequently, the PAPR optimization can be confined to the interval $c_{2} \in [0, \frac{1}{2})$.
\end{proof}

From a signal processing perspective, the IDAFT in the Agile-AFDM system constitutes a unitary transformation that preserves the statistical properties of the signal.
For instance, if the input symbols $\mathbf{x}$ follow a complex Gaussian distribution, the transformed signal $\bm{s}$ will also exhibit a complex Gaussian distribution, maintaining the same second-order statistics. While the parameter $c_2$ influences the instantaneous power distribution of the transmitted signal, the statistical invariance property of unitary transformations ensures that the long-term average PAPR performance, characterized by $\mathbb{E}_{\mathbf{x}}[\xi(\mathbf{c},\mathbf{x})]$, remains invariant to the specific choice of $c_2$.

This statistical invariance reveals a fundamental limitation of fixed-parameter AFDM systems: while they maintain average PAPR performance, they cannot adapt to the instantaneous PAPR characteristics of individual data blocks. This underscores the critical necessity for dynamic parameter selection, where effective PAPR reduction is achieved by optimizing $c_2$ on a per-block basis as per the specific data realization $\mathbf{x}$.

\subsection{PAPR Reduction}
Leveraging the periodicity established in Lemma \ref{thm:AFDM_PAPR_periodicity}, we confine the PAPR optimization to the compact interval $c_{2} \in [0, \frac{1}{2})$. The optimization objective is formulated as
\begin{equation}\label{eq:PAPR_optimization}
c_{2}^* = \arg\min_{c_{2} \in [0,\frac{1}{2})} \max_{t \in [0,T_{\text{sym}}]} g(t),
\end{equation}
where minimizing the peak value of $g(t)$ directly corresponds to PAPR minimization.

In practical Agile-AFDM implementations, the PAPR demonstrates significant sensitivity to fine variations in $c_{2}$, necessitating a carefully calibrated search to capture the nuances in PAPR behavior. The primary challenge lies in managing the $\max{\left\{g(t)\right\}}$ operation. While theoretically an infinite norm, the $\max$ operation can be effectively approximated by the $n$-th root of the integral of $|g(t)|^n$ over time $T_{\text{sym}}$, $\sqrt[n]{\int_{0}^{T_{\text{sym}}}|g(t)|^n dt}$, even when $n$ is not particularly large \cite{closedform}. In this light, we introduce a surrogate optimization function:
\begin{equation}\label{eq:AFDM_surrogate}
I = \int_0^{T_{\text{sym}}} g(t)^4 dt.
\end{equation}

To locate the optimal $c_{2}$ that minimizes PAPR, we compute the derivative of the surrogate function
\begin{equation}\label{eq:AFDM_derivative}
I'(c_{2}) = \frac{d}{dc_{2}}\int_0^{T_{\text{sym}}} g(t)^4 dt
= 4 \int_0^{T_{\text{sym}}} g(t)^3 \cdot \frac{\partial g(t)}{\partial c_{2}} dt,
\end{equation}
where the derivative of the oscillatory component is given by
\begin{equation}\label{eq:dg_dc2}
\begin{aligned}
&\frac{\partial g(t)}{\partial c_{2}} \!\!=\!\! \sum_{p=1}^{N-1} \!\rho_{p}^{(1)}(c_{2})\!\cos\left(\frac{2\pi pt}{T_{\text{sym}}}\right) \!+\! \sum_{p=1}^{N-1} \!\rho_{p}^{(2)}(c_{2})\!\cos\left(\frac{2\pi pt}{T_{\text{sym}}}\right) \\
&+ \sum_{p=1}^{N-1} \rho_{p}^{(3)}(c_{2})\sin\left(\frac{2\pi pt}{T_{\text{sym}}}\right) + \sum_{p=1}^{N-1} \rho_{p}^{(4)}(c_{2})\sin\left(\frac{2\pi pt}{T_{\text{sym}}}\right),
\end{aligned}
\end{equation}
with the derivative coefficients defined as
\begin{equation*}
\rho_{p}^{(1)}(c_{2})\triangleq -\sum\limits_{m=0}^{N-1-p}p(2m+p)\lambda_{m,p}\sin\beta_{m,p}(c_{2}),
\end{equation*}
\begin{equation*}
\rho_{p}^{(2)}(c_{2})\triangleq-\sum\limits_{m=0}^{N-1-p}p(2m+p)\mu_{m,p}\cos\beta_{m,p}(c_{2}),
\end{equation*}
\begin{equation*}
\rho_{p}^{(3)}(c_{2})\triangleq-\sum\limits_{m=0}^{N-1-p}p(2m+p)\lambda_{m,p}\cos\beta_{m,p}(c_{2}),
\end{equation*}
\begin{equation*}
\rho_{p}^{(4)}(c_{2})\triangleq\sum\limits_{m=0}^{N-1-p}p(2m+p)\mu_{m,p}\sin\beta_{m,p}(c_{2}).
\end{equation*}

Upon first examination, \eqref{eq:AFDM_derivative} appears computationally intensive due to its inclusion of $256\left(N-1\right)^4$ individual integrals, all of which are initially perceived as complex due to their structure. However, these integrals can be efficiently computed by leveraging the orthogonality properties of trigonometric functions.

\begin{thm}\label{thm:AFDM_trigonometric}
Let $\omega_1(t), \omega_2(t), \omega_3(t), \omega_4(t) \in \{\cos(t), \sin(t)\}$.
For indices $1\le k,l,m,n\le N-1$, the integral of the product of these trigonometric functions over a full period is given by
\begin{equation*}
\int_0^{2\pi} \omega_1(kt)\omega_2(lt)\omega_3(mt)\omega_4(nt)\, dt =  \frac{\pi}{4} \bm{q}_i
\begin{bmatrix}
\begin{smallmatrix}
\delta(k + l + m + n)\\
\delta(k + l - m - n)\\
\delta(k + l + m - n)\\
\delta(k + l - m + n)\\
\delta(k - l + m + n)\\
\delta(k - l - m - n)\\
\delta(k - l + m - n)\\
\delta(k - l - m + n)
\end{smallmatrix}
\end{bmatrix},
\end{equation*}
where $\bm{q}_i$ denotes the $i^{th}$ row of matrix $\bm{Q}$:
\begin{equation*}\label{Q_AFDM}
\bm{Q} = \left[
\begin{smallmatrix}
0 & +1 & +1 & +1 & +1 & +1 & +1 & +1 \\
0 & 0 & 0 & 0 & 0 & 0 & 0 & 0 \\
0 & 0 & 0 & 0 & 0 & 0 & 0 & 0 \\
0 & -1 & +1 & +1 & -1 & -1 & +1 & +1 \\
0 & 0 & 0 & 0 & 0 & 0 & 0 & 0 \\
0 & +1 & +1 & -1 & +1 & -1 & -1 & +1 \\
0 & +1 & -1 & +1 & +1 & -1 & +1 & -1 \\
0 & 0 & 0 & 0 & 0 & 0 & 0 & 0 \\
0 & 0 & 0 & 0 & 0 & 0 & 0 & 0 \\
0 & +1 & +1 & -1 & -1 & +1 & +1 & -1 \\
0 & +1 & -1 & +1 & -1 & +1 & -1 & +1 \\
0 & 0 & 0 & 0 & 0 & 0 & 0 & 0 \\
0 & -1 & -1 & -1 & +1 & +1 & +1 & +1 \\
0 & 0 & 0 & 0 & 0 & 0 & 0 & 0 \\
0 & 0 & 0 & 0 & 0 & 0 & 0 & 0 \\
0 & +1 & -1 & -1 & -1 & -1 & +1 & +1
\end{smallmatrix}
\right],
\end{equation*}
and index $i$ is determined by $i = 8\delta\big(\omega_1(t)\!=\!\sin(t)\big) + 4\delta\big(\omega_2(t)\!=\!\sin(t)\big) +  2\delta\big(\omega_3(t)\!=\!\sin(t)\big) + \delta\big(\omega_4(t)\!=\!\sin(t)\big) + 1$.
\end{thm}

\begin{proof}
    See Appendix \ref{sec:AppB}.
\end{proof}

\begin{algorithm}[t]
\caption{Agile-AFDM for PAPR reduction.}\label{algo:AFDM_optimization}
\begin{algorithmic}[1]
\State {\bf Input:} $N$, and data symbol vector $\mathbf{x}$.
\State {\bf Output:} Optimal chirp parameter $c_{2}^{*}$.

\State Determine the initial search set $\Omega_0$ based on Lemma \ref{thm:AFDM_PAPR_periodicity}: $c_{2} \in [0, \frac{1}{2})$.

\For{$c_{2}^{(i)} \in \Omega_0$}
    \State Compute $\gamma_{p}^{(j)}(c_{2}^{(i)})$ for $j = 1, 2, 3, 4$.
    \State Compute $I'(c_{2}^{(i)})$ by Theorem~\ref{thm:AFDM_trigonometric}.
\EndFor

\State Initialize $\Omega = \emptyset$.
\For{$c_{2}^{(i)}\in\Omega_0$}
    \If{$I'(c_{2}^{(i)}) \leq 0$ and $I'(c_{2}^{(i)} + \Delta c) \geq 0$} 
        \State $\Omega = \Omega \cup \left\{c_{2}^{(i)} + j \Delta c^{\prime} : j = 0, 1, 2, \dots, \frac{\Delta c}{\Delta c^{\prime}}\right\}$.
    \EndIf
\EndFor
 
\For{$c_{2}^{(j)}\in\Omega$}
    \If{$I'(c_{2}^{(j)}) \leq 0$ and $I'(c_{2}^{(j)} + \Delta c^{\prime}) \geq 0$}
        \State Retain $c_{2}^{(j)}$ in $\Omega$.
    \Else
        \State $\Omega=\Omega\backslash\{c_{2}^{(j)}\}$.
    \EndIf
\EndFor

\State Search through $\Omega$: $c_{2}^*=\arg\min_{c_{2}\in\Omega}\xi(c_{2})$.
\end{algorithmic}
\end{algorithm}

Theorem~\ref{thm:AFDM_trigonometric} provides an efficient method for computing $I'(c_{2})$, which underpins the design of an algorithm to identify a refined set of candidate chirp parameters, $\Omega$, for locating the optimal $c_{2}^*$ that minimizes PAPR. The algorithm is described in Algorithm~\ref{algo:AFDM_optimization}, with details on its steps outlined below.

The algorithm begins by taking two input parameters: the number of subcarriers $N$, and the data symbol vector $\mathbf{x}$. Using these inputs, it establishes an initial search range for $c_{2}$ based on Lemma \ref{thm:AFDM_PAPR_periodicity}: $c_{2} \in [0, \frac{1}{2})$. This range is then discretized with an initial step size $\Delta c$, generating an initial search set:
\begin{equation}\label{eq:AFDM_init}
\Omega_0 = \left\{i\Delta c : i = 0, 1, 2, \ldots, \lfloor \frac{1}{2\Delta c} \rfloor - 1 \right\}.
\end{equation}

For each $c_{2}^{(i)} \in \Omega_0$, the algorithm computes $\gamma_{p}^{(j)}(c_{2}^{(i)})$ for $j = 1, 2, 3, 4$, along with the corresponding $I'(c_{2}^{(i)})$ using Theorem~\ref{thm:AFDM_trigonometric}. It then identifies potential local minima by verifying if $I'(c_{2}^{(i)}) \leq 0$ and $I'(c_{2}^{(i)} + \Delta c) \geq 0$. When these conditions are met, indicating a local minimum, we form a finer set of candidate values:
\begin{equation*}
\Omega = \Omega \cup \left\{c_{2}^{(i)} + j\Delta c^{\prime} : j = 0, 1, 2, \ldots, \frac{\Delta c}{\Delta c^{\prime}}\right\},
\end{equation*}
where $\Delta c^{\prime}$ is a finer step size, and $\Omega$ is initialized as an empty set.

Next, the algorithm further narrows the finer search set $\Omega$ by iterating over each candidate $c_{2}^{(j)}\in \Omega$ and retaining only those values that satisfy $I'(c_{2}^{(j)}) \leq 0$ and $I'(c_{2}^{(j)} + \Delta c^{\prime}) \geq 0$.

In the final step, the algorithm searches over $\Omega$ to identify the optimal chirp parameter that minimizes the PAPR.

To summarize, Agile-AFDM provides a new approach to PAPR reduction in multi-carrier systems. At its core, the method leverages data-aware adaptation of the chirp parameter, enabling dynamic, block-specific PAPR minimization. 
While PAPR is notoriously analytically intractable, the proposed optimization framework, underpinned by the derived surrogate function and efficient search algorithm, provides a practical and effective solution. The resultant performance gains in PAPR reduction will be quantified and compared against benchmark schemes via simulation results in Section \ref{Sec:Simu}.

%% file: SIR.tex
The preservation of subcarrier orthogonality stands as the cornerstone of reliable communication in high-mobility scenarios \cite{benzine_affine_2024_spawc}. While traditional static AFDM configurations provide robustness against worst-case channel conditions, they fundamentally lack the capability to adapt to the instantaneous realizations of multipath components and Doppler shifts that characterize practical wireless environments. This limitation becomes particularly acute in the context of ICI \cite{sec2_ici}, where fixed parameters cannot dynamically reconfigure the effective channel matrix $\bm{H}_{\text{eff}}^{\text{comm}}$ to suppress interference arising from specific channel realizations.

In this section, we introduce a paradigm shift from static robustness to dynamic interference suppression through Agile-AFDM. Our approach is built upon two foundational insights: first, that the structure of $\bm{H}_{\text{eff}}^{\text{comm}}$, and consequently the ICI, is jointly orchestrated by both chirp parameters $c_1$ and $c_2$; and second, that optimal parameter selection should be conditioned not only on instantaneous CSI but also on the transmitted data symbols $\bm{x}$ themselves. This data-aware optimization represents a novel philosophical approach to waveform design that extends far beyond AFDM to any parameterized modulation scheme.

\subsection{ICI Analysis and SIR}
The core objective of this section is to develop a framework for dynamically optimizing $c_1$ and $c_2$ to reconfigure the effective channel matrix, thereby reducing the ICI. The detrimental effects of ICI manifest directly in the DAF-domain received symbols. For the $p$-th subcarrier, the input-output relationship can be decomposed into three terms, as given in \eqref{eq:commy_subcarrier}, where the ICI term arises from the off-diagonal elements of $\bm{H}_{\text{eff}}^{\text{comm}}$.

To quantify and suppress ICI, we define the SIR $\zeta$ as a comprehensive metric across all subcarriers:
\begin{equation}
\zeta = \frac{1}{N} \sum_{p=0}^{N-1} \frac{ P_{\text{sig},p} }{ P_{\text{int},p} },~~~
\zeta_{\text{dB}}\triangleq 10\log_{10}\zeta
\end{equation}
where the signal and interference powers are defined as
\begin{align*}
P_{\text{sig},p} &\triangleq \left| H_{\text{eff}}^{\text{comm}}[p,p] x[p] \right|^2, \\
P_{\text{int},p} &\triangleq \left| \sum_{\substack{q=0 \\ q \neq p}}^{N-1} H_{\text{eff}}^{\text{comm}}[p,q] x[q] \right|^2.
\end{align*}

The fundamental advantage of dynamic parameter configuration in Agile-AFDM systems is captured by the following theoretical result, which quantifies the performance gain achievable through adaptive optimization.

\begin{thm} \label{thm:sir_dynamic}
Let $\mathcal{C}$ be a finite set of feasible configuration parameters, and let $\bm{x} \in \mathcal{X}$ be a random vector defined on the probability space $(\mathcal{W}, \mathcal{G}, \mathbb{P})$. Let $\zeta: \mathcal{C} \times \mathcal{X} \to \mathbb{R}$ be a measurable objective function representing the system SIR.

Define the optimal static parameter as $\bm{c}^* \in \arg\max_{\bm{c} \in \mathcal{C}} \mathbb{E}_{\bm{x}} [\zeta(\bm{c}, \bm{x})]$. The following inequality holds:
\begin{equation}
\mathbb{E}_{\bm{x}} \left[ \max_{\bm{c} \in \mathcal{C}} \zeta(\bm{c}, \bm{x}) \right] \geq \max_{\bm{c} \in \mathcal{C}} \mathbb{E}_{\bm{x}} \left[ \zeta(\bm{c}, \bm{x}) \right].
\end{equation}
Furthermore, the inequality is strict if and only if
\begin{equation}
\mathbb{P} \left( \left\{ \bm{x} \in \mathcal{X} : \zeta(\bm{c}^*, \bm{x}) < \max_{\bm{c} \in \mathcal{C}} \zeta(\bm{c}, \bm{x}) \right\} \right) > 0.
\end{equation}
In simpler terms, the strict inequality holds provided that the parameter $\bm{c}^*$, which is optimal on average, is strictly suboptimal for a set of channel realizations with non-zero probability measure.
\end{thm}

\begin{proof}
We proceed by establishing the weak inequality first, followed by the condition for strictness.

\textit{Weak inequality:}
For any fixed parameter vector $\bm{c}' \in \mathcal{C}$, and for any realization $\bm{x} \in \mathcal{X}$, the value of the objective function is bounded above by its maximum over the set $\mathcal{C}$:
$$
\zeta(\bm{c}', \bm{x}) \leq \max_{\bm{c} \in \mathcal{C}} \zeta(\bm{c}, \bm{x}).
$$
Since the expectation operator $\mathbb{E}_{\bm{x}}[\cdot]$ is monotonic (i.e., if $A \leq B$ almost surely, then $\mathbb{E}[A] \leq \mathbb{E}[B]$), taking the expectation of both sides yields:
$$
\mathbb{E}_{\bm{x}} \left[ \zeta(\bm{c}', \bm{x}) \right] \leq \mathbb{E}_{\bm{x}} \left[ \max_{\bm{c} \in \mathcal{C}} \zeta(\bm{c}, \bm{x}) \right].
$$
Since this inequality holds for \textit{any} arbitrary $\bm{c}' \in \mathcal{C}$, it must also hold for the specific parameter $\bm{c}^*$ that maximizes the right-hand term (the expected SIR). Therefore:
$$
\max_{\bm{c} \in \mathcal{C}} \mathbb{E}_{\bm{x}} \left[ \zeta(\bm{c}, \bm{x}) \right] = \mathbb{E}_{\bm{x}} \left[ \zeta(\bm{c}^*, \bm{x}) \right] \leq \mathbb{E}_{\bm{x}} \left[ \max_{\bm{c} \in \mathcal{C}} \zeta(\bm{c}, \bm{x}) \right].
$$

\textit{Strict inequality:}
Let us define a non-negative random variable $\Delta(\bm{x})$ representing the performance gap between the adaptive strategy and the optimal static strategy:
$$
\Delta(\bm{x}) \triangleq \max_{\bm{c} \in \mathcal{C}} \zeta(\bm{c}, \bm{x}) - \zeta(\bm{c}^*, \bm{x}).
$$
By definition, $\Delta(\bm{x}) \geq 0$ for all $\bm{x}$. The expectation of the performance gap is:
$$
\mathbb{E}_{\bm{x}}[\Delta(\bm{x})] = \mathbb{E}_{\bm{x}} \left[ \max_{\bm{c} \in \mathcal{C}} \zeta(\bm{c}, \bm{x}) \right] - \mathbb{E}_{\bm{x}} \left[ \zeta(\bm{c}^*, \bm{x}) \right].
$$
For the inequality to be strict (i.e., $\mathbb{E}_{\bm{x}}[\Delta(\bm{x})] > 0$), the random variable $\Delta(\bm{x})$ must not be equal to zero almost surely. This requires that there exists a set of realizations with positive probability where $\Delta(\bm{x}) > 0$.

Mathematically, if $\mathbb{P}(\Delta(\bm{x}) > 0) > 0$, then $\mathbb{E}_{\bm{x}}[\Delta(\bm{x})] > 0$. This condition implies that for some non-negligible set of channel states, the statically optimal parameter $\bm{c}^*$ fails to achieve the instantaneous maximum SIR. Conversely, if $\bm{c}^*$ is optimal for all $\bm{x}$ (almost surely), the gap is zero and equality holds. 

Thus, under the condition that no single parameter configuration is universally optimal, adaptive parameter selection yields a strictly higher expected SIR.
\end{proof}

\subsection{SIR Maximization via Fractional Programming (FP)}
The SIR performance is governed by the chirp parameters $c_1$ and $c_2$ through their influence on the effective channel matrix elements $H_{\text{eff}}^{\text{comm}}[p,q]$. Achieving optimal SIR requires solving the following optimization problem:
\begin{equation} \label{eq:sir_optimization}
\{c_1^*, c_2^*\} = \underset{c_1, c_2 \in [0,1)}{\arg\max} \ \zeta(c_1, c_2),
\end{equation}
where $\zeta(c_1, c_2)$ denotes the system SIR as a function of the chirp parameters.

\begin{algorithm}[t]
\caption{Agile-AFDM for SIR maximization.}
\label{alg:block_wise_adam}
\begin{algorithmic}[1]

\State \textbf{Input:} Channel parameters $\{l_i, \nu_i, h_i\}_{i=1}^P$, 
\State \hspace{0.95cm} convergence tolerance $\epsilon_{\text{con}}$, 
\State \hspace{0.95cm} parameter space division $B_1 \times B_2$,
\State \hspace{0.95cm} maximum iterations $I_{\text{max}}$,
\State \hspace{0.95cm} number of subcarriers $N$, channel realizations $N_h$,
\State \hspace{0.95cm} number of data blocks $N_b$,
\State \hspace{0.95cm} Adam hyperparameters $\alpha$, $\beta_1$, $\beta_2$, $\epsilon_{\text{adam}}$
\State \textbf{Output:} Optimal chirp parameters $\{c_1^*, c_2^*\}$, 
\State \hspace{1.2cm} achieved SIR $\zeta^*$

\State $\zeta^* \leftarrow -\infty$, $\bm{c}^* \leftarrow [0, 0]$
\State Partition parameter space $[0,1)^2$ into $B_1 \times B_2$ blocks

\For{each data block in parallel} \label{line:block_loop}
    \For{each parameter block $b = 1$ to $B_1 \times B_2$ in parallel} \label{line:param_block_loop}
        \State $\bm{c}^{(0)} \leftarrow \text{grid point of block } b$, $z_p^{(0)} \leftarrow 1$, $k_{\text{fp}} \leftarrow 0$
        
        \Repeat \label{line:fp_repeat}            
            \State Update auxiliary variables $\{z_p^{(k_{\text{fp}}+1)}\}_{p=0}^{N-1}$
            
            \State $\bm{m} \leftarrow \bm{0}$, $\bm{v} \leftarrow \bm{0}$, $t \leftarrow 0$, $\mathrm{iter} \leftarrow 0$
            
            \Repeat \label{line:adam_repeat}
                \State Compute $f(\bm{c})$ using \eqref{eq:quadratic_transform}
                \State $\bm{g} \leftarrow \nabla_{\bm{c}} f(\bm{c}^{(\mathrm{iter})})$
                \State Update Adam: $\bm{m}, \bm{v}, t$,
                \State and apply bias correction
                \State $\bm{c}^{(\mathrm{iter}+1)} \leftarrow \bm{c}^{(\mathrm{iter})} + \alpha\hat{\bm{m}} \oslash (\sqrt{\hat{\bm{v}}} + \epsilon_{\text{adam}})$
                \State $\bm{c}^{(\mathrm{iter}+1)} \leftarrow \text{mod}(\bm{c}^{(\mathrm{iter}+1)}, 1)$, 
                \State $\mathrm{iter} \leftarrow \mathrm{iter} + 1$
            \Until{$\|\bm{c}^{(\mathrm{iter})} - \bm{c}^{(\mathrm{iter}-1)}\| < \epsilon_{\text{con}}$ or $\mathrm{iter} > I_{\text{max}}$} \label{line:adam_until}
            
            \State $\bm{c}^{(k_{\text{fp}}+1)} \leftarrow \bm{c}^{(\mathrm{iter})}$, $k_{\text{fp}} \leftarrow k_{\text{fp}} + 1$
        \Until{$\|\bm{c}^{(k_{\text{fp}})} - \bm{c}^{(k_{\text{fp}}-1)}\| < \epsilon_{\text{con}}$ or $k_{\text{fp}} > I_{\text{max}}$} \label{line:fp_until}
        
        \State $\zeta_{\text{current}} \leftarrow \frac{1}{N} \sum_{p=0}^{N-1} \frac{P_{\text{sig},p}(\bm{c}^{(k_{\text{fp}})})}{P_{\text{ICI},p}(\bm{c}^{(k_{\text{fp}})}) + \delta}$
        \If{$\zeta_{\text{current}} > \zeta^*$}
            \State $\zeta^* \leftarrow \zeta_{\text{current}}$, $\bm{c}^* \leftarrow \bm{c}^{(k_{\text{fp}})}$
        \EndIf
    \EndFor \label{line:param_block_end}
\EndFor \label{line:block_end}

\end{algorithmic}
\end{algorithm}

The SIR optimization problem in \eqref{eq:sir_optimization} exhibits a highly non-convex landscape with multiple local optima, posing significant challenges for direct optimization. To address this, we propose a block-wise alternating optimization framework that integrates FP \cite{stancu2012fractional} to handle the sum-of-ratios structure with the adaptive moment estimation (Adam) optimizer \cite{kingma2014adam} for efficient parameter updates. The core SIR maximization can be formulated as a sum-of-ratios FP problem:
\begin{equation}
\label{eq:sum_ratios}
\max_{c_1, c_2} \sum_{p=0}^{N-1} \frac{P_{\text{sig},p}(c_1, c_2)}{P_{\text{ICI},p}(c_1, c_2)},
\end{equation}
which remains inherently non-convex. 

By applying the quadratic transform \cite{shen2018fractional}, we introduce auxiliary variables $\bm{z} = \{z_p\}_{p=0}^{N-1}$ to reformulate this into an equivalent surrogate problem:
\begin{equation}
\label{eq:quadratic_transform}
\max_{c_1, c_2, \bm{z}} \sum_{p=0}^{N-1} 
\left( 
2z_p\sqrt{P_{\text{sig},p}(c_1, c_2)} - z_p^2 P_{\text{ICI},p}(c_1, c_2)
\right).
\end{equation}
This transformation effectively decouples the numerator and denominator of each ratio, enabling efficient optimization through alternating updates between the auxiliary variables $\bm{z}$ and the chirp parameters $(c_1, c_2)$.

To mitigate the effects of local optima, we partition the parameter space $[0,1)^2$ into $B_1 \times B_2$ uniform blocks. Within each block, we initialize the search from uniformly distributed grid points to thoroughly explore the parameter space and minimize the risk of premature convergence to suboptimal local optima. The Adam optimizer then optimizes the surrogate problem through alternating updates. The gradient of the surrogate objective $f(\bm{c}, \bm{z}) = \sum_{p=0}^{N-1} \left(2z_p\sqrt{P_{\text{sig},p}(\bm{c})} - z_p^2 P_{\text{ICI},p}(\bm{c})\right)$ with respect to $\bm{c} = [c_1, c_2]$ is computed numerically via central differences:
\begin{equation*}
\frac{\partial f}{\partial c_i} \approx \frac{f(c_i + \Delta_c, \bm{z}) - f(c_i - \Delta_c, \bm{z})}{2\Delta_c}, \quad i = 1,2,
\end{equation*}
where $\Delta_c$ is a small perturbation (e.g., $10^{-6}$) ensuring numerical stability.

The complete optimization algorithm is summarized in Algorithm \ref{alg:block_wise_adam}, which proceeds as follows. Given the input channel parameters, convergence tolerance $\epsilon_{\text{con}}$, maximum iterations $I_{\text{max}}$, parameter space division $B_1 \times B_2$, number of data blocks $N_b$, and Adam hyperparameters $\alpha$, $\beta_1$, $\beta_2$, $\epsilon_{\text{adam}}$, we initialize the global best SIR $\zeta^* \leftarrow -\infty$ and optimal parameters $\bm{c}^* \leftarrow [0,0]$. The parameter space $[0,1)^2$ is partitioned into $B_1 \times B_2$ uniform blocks $\{\mathcal{B}_b\}_{b=1}^{B_1B_2}$. Each data block is then processed in parallel: for each parameter block $\mathcal{B}_b$, we initialize parameters $\bm{c}^{(0)}$ at the corresponding grid point and initialize Adam variables (momentum $\bm{m}^{(0)} \leftarrow \bm{0}$, variance $\bm{v}^{(0)} \leftarrow \bm{0}$, iteration counter $t \leftarrow 0$, internal iteration $\mathrm{iter} \leftarrow 0$). The alternating update loop proceeds as follows: auxiliary variables $\bm{z}$ are first updated for each subcarrier $p$ as
\begin{equation}
z_p^{(\mathrm{iter}+1)} = \frac{\sqrt{P_{\text{sig},p}(\bm{c}^{(\mathrm{iter})})}}{P_{\text{ICI},p}(\bm{c}^{(\mathrm{iter})}) + \delta},
\end{equation}
where $\delta$ is a regularization term preventing division by zero. The gradient $\bm{g}^{(\mathrm{iter})} = \nabla_{\bm{c}} f(\bm{c}^{(\mathrm{iter})}, \bm{z}^{(\mathrm{iter}+1)})$ is then computed using central differences. Adam's momentum and variance are updated by incrementing $t$:
\begin{equation*}
\bm{m}^{(\mathrm{iter})} = \beta_1 \bm{m}^{(\mathrm{iter}-1)} + (1-\beta_1)\bm{g}^{(\mathrm{iter})},
\end{equation*}
\begin{equation*}
\bm{v}^{(\mathrm{iter})} = \beta_2 \bm{v}^{(\mathrm{iter}-1)} + (1-\beta_2)(\bm{g}^{(\mathrm{iter})} \odot \bm{g}^{(\mathrm{iter})}),
\end{equation*}
with hyperparameters $\beta_1$ and $\beta_2$. Bias correction is applied:
\begin{equation*}
\hat{\bm{m}} = \frac{\bm{m}^{(\mathrm{iter})}}{1-\beta_1^t}, \quad \hat{\bm{v}} = \frac{\bm{v}^{(\mathrm{iter})}}{1-\beta_2^t},
\end{equation*}
followed by parameter updates with adaptive step size:
\begin{equation*}
\bm{c}^{(\mathrm{iter}+1)} = \bm{c}^{(\mathrm{iter})} + \alpha \cdot \frac{\hat{\bm{m}}}{\sqrt{\hat{\bm{v}}} + \epsilon_{\text{adam}}},
\end{equation*}
where $\alpha$ is the learning rate and $\epsilon_{\text{adam}}$ ensures numerical stability. Periodicity is enforced via $\bm{c}^{(\mathrm{iter}+1)} = \text{mod}(\bm{c}^{(\mathrm{iter}+1)}, 1)$ to maintain parameters within $[0,1)$. Internal iterations terminate when $\|\bm{c}^{(\mathrm{iter})} - \bm{c}^{(\mathrm{iter}-1)}\| < \epsilon_{\text{con}}$ or $\mathrm{iter} > I_{\text{max}}$. The SIR for the converged parameters is computed as:
\begin{equation*}
\zeta_{\text{current}} = \frac{1}{N} \sum_{p=0}^{N-1} \frac{P_{\text{sig},p}(\bm{c}^{(\mathrm{iter})})}{P_{\text{ICI},p}(\bm{c}^{(\mathrm{iter})}) + \delta},
\end{equation*}
and the global optimum is updated if $\zeta_{\text{current}} > \zeta^*$. 

Finally, the algorithm returns the globally optimal parameters $\bm{c}^* = [c_1^*, c_2^*]$ and the maximum achieved SIR $\zeta^*$.
\begin{rem}
   The computational complexity of Algorithm~\ref{alg:block_wise_adam} can be analyzed as follows. 
   The algorithm processes $N_b$ data blocks and $B_1 \times B_2$ parameter blocks in parallel. Each parameter block undergoes an FP outer loop with at most $I_{\text{max}}$ iterations, and each FP iteration invokes an Adam optimizer running for $I_{\text{max}}$ iterations. 
   
   The objective function evaluation per Adam iteration requires computing $P_{\text{sig},p}$ and $P_{\text{ICI},p}$ for all $N$ subcarriers, with costs of $\mathcal{O}(P)$ and $\mathcal{O}(NP)$ per subcarrier, respectively. This leads to an overall complexity of $\mathcal{O}(N^2 P)$ per function evaluation. Averaging over $N_h$ channel realizations gives $\mathcal{O}(N_h N^2 P)$. Considering the gradient computations, the overall complexity becomes $\mathcal{O}(N_b B_1 B_2 I_{\text{max}}^2 N_h N^2 P)$. 
   
   In contrast, exhaustive grid search over a discretization grid with resolution $G$ in each dimension requires $\mathcal{O}(G^2 \cdot N_h N^2 P)$ evaluations, which becomes prohibitive for fine-grained parameter space exploration. The proposed block-wise optimization approach offers significant computational advantage through its parallel processing and adaptive convergence, particularly when $B_1 B_2 I_{\text{max}}^2 \ll G^2$. 
\end{rem}

Overall, our optimization framework enables Agile-AFDM to dynamically reconfigure chirp parameters for each transmission block, achieving substantial ICI suppression and communication reliability gains, as demonstrated in Section \ref{Sec:Simu}.

%% file: CRLB.tex
The convergence of communication and sensing in next-generation wireless systems imposes stringent dual requirements on the physical-layer waveform \cite{ISACsurvey}. Beyond ensuring reliable data transmission, the waveform must also function as a precise radar probe, capable of high-accuracy estimation of target parameters such as range and velocity. The fundamental limit of this estimation accuracy is quantified by the CRLB \cite{mendel1995lessons}, which provides a lower bound on the variance of any unbiased estimator. While AFDM has demonstrated inherent robustness for sensing in doubly-dispersive channels, existing studies have largely evaluated its performance with static chirp parameters, designed for worst-case channel conditions. This static approach fails to exploit a critical degree of freedom: the potential to dynamically tailor the waveform's time-frequency structure to the instantaneous sensing scenario, thereby minimizing estimation uncertainty.

In this section, we unlock the full potential of Agile-AFDM for high-precision sensing. We move beyond the static paradigm by formulating and solving a dynamic, per-block optimization problem: minimizing the CRLB for delay and Doppler estimation through real-time adaptation of the chirp parameters $(c_1, c_2)$. This transforms AFDM from a robust sensing waveform into an intelligent one, capable of self-optimizing its estimation performance.

\subsection{CRLB Analysis}

To establish the foundation for chirp adaptation, we first derive a closed-form expression of the CRLB for Agile-AFDM, explicitly accounting for the influence of the chirp parameters. Based on the sensing signal model in Section~\ref{Sec:System Model}, our goal is to estimate the target's delay $l$ and normalized Doppler shift $\nu$ from the received echo $\mathbf{y}^{\text{sens}}$ in \eqref{eq:sensingy}.

The unknown parameter vector includes both the parameters of interest $[l, \nu]^T\triangleq \boldsymbol{\theta}$ and the complex reflection coefficient $\alpha$, which is treated as a nuisance parameter:
\begin{equation}
\boldsymbol{\eta} \triangleq [\boldsymbol{\theta}^T, \boldsymbol{\kappa}^T]^T = [l, \nu, \alpha_{\mathrm{R}}, \alpha_{\mathrm{I}}]^T,
\end{equation}
where $\boldsymbol{\kappa} \triangleq [\alpha_{\mathrm{R}}, \alpha_{\mathrm{I}}]^T$ represents the real and imaginary parts of $\alpha$.

Given the input-output relationship in \eqref{eq:sensingy}, the received echo in the DAF domain is expressed as
\begin{equation*}
\bm{y}^{\text{sens}} = \boldsymbol{\chi}(\boldsymbol{\eta}) + \tilde{\bm{w}}^{\text{sens}},
\end{equation*}
where $\boldsymbol{\chi}(\boldsymbol{\eta}) \triangleq \bm{H}_{\text{eff}}^{\text{sens}} \bm{x}$ represents the noise-free signal vector.
To isolate the parameters of interest, we decompose the effective channel as $\bm{H}_{\text{eff}}^{\text{sens}} = \alpha \bm{G}(l, \nu)$. Consequently, the mean vector becomes
\begin{equation*}
\boldsymbol{\chi}(\boldsymbol{\eta}) = \alpha \underbrace{\bm{G}(l, \nu) \bm{x}}_{\bm{u}(l, \nu)} = \alpha \bm{u},
\end{equation*}
where $\bm{u}$ denotes the normalized reference signal vector, which depends on $l$ and $\nu$ but is independent of the reflection coefficient $\alpha$.

For a parameter vector estimated from observations in complex AWGN, the elements of the Fisher Information Matrix (FIM) $\mathbf{J} \in \mathbb{R}^{4 \times 4}$ follow the standard form\cite{mcrlb2}:
\begin{equation} \label{eq:FIM_def}
[\mathbf{J}]_{m,n} = \frac{2}{N_0} \Re\left\{ \left( \frac{\partial \boldsymbol{\chi}}{\partial \eta_m} \right)^H \frac{\partial \boldsymbol{\chi}}{\partial \eta_n} \right\}.
\end{equation}

The partial derivatives of $\boldsymbol{\chi}$ with respect to the parameters of interest and the nuisance parameters are derived as:
\begin{subequations}
\begin{align*}
\frac{\partial \boldsymbol{\chi}}{\partial l} &= \alpha \frac{\partial \bm{u}}{\partial l} \triangleq \alpha \bm{u}_l, \quad
\frac{\partial \boldsymbol{\chi}}{\partial \nu} = \alpha \frac{\partial \bm{u}}{\partial \nu} \triangleq \alpha \bm{u}_\nu, \\
\frac{\partial \boldsymbol{\chi}}{\partial \alpha_{\mathrm{R}}} &= \bm{u}, \quad
\frac{\partial \boldsymbol{\chi}}{\partial \alpha_{\mathrm{I}}} = j \bm{u}.
\end{align*}
\end{subequations}
Substituting these derivatives into \eqref{eq:FIM_def}, we can partition the FIM into blocks:
\begin{equation}
\mathbf{J} = \begin{bmatrix}
\mathbf{J}_{\boldsymbol{\theta \theta}} & \mathbf{J}_{\boldsymbol{\theta \kappa}} \\
\mathbf{J}_{\boldsymbol{\kappa \theta}} & \mathbf{J}_{\boldsymbol{\kappa \kappa}}
\end{bmatrix}.
\end{equation}
Specifically, the sub-matrices are
\begin{equation*}
\mathbf{J}_{\boldsymbol{\theta \theta}} = \frac{2|\alpha|^2}{N_0} \begin{bmatrix}
\|\bm{u}_l\|^2 & \Re\{\langle \bm{u}_l, \bm{u}_\nu \rangle\} \\
\Re\{\langle \bm{u}_l, \bm{u}_\nu \rangle\} & \|\bm{u}_\nu\|^2
\end{bmatrix},
\end{equation*}
\begin{equation*}
\mathbf{J}_{\boldsymbol{\kappa \kappa}} = \frac{2}{N_0} \|\bm{u}\|^2 \mathbf{I}_2,
\end{equation*}
\begin{equation*}
\mathbf{J}_{\boldsymbol{\theta \kappa}} = \mathbf{J}_{\boldsymbol{\kappa \theta}}^T = \frac{2}{N_0} \begin{bmatrix}
\Re\{(\alpha)^* \langle \bm{u}_l, \bm{u} \rangle\} & \Im\{(\alpha)^* \langle \bm{u}_l, \bm{u} \rangle\} \\
\Re\{(\alpha)^* \langle \bm{u}_\nu, \bm{u} \rangle\} & \Im\{(\alpha)^* \langle \bm{u}_\nu, \bm{u} \rangle\}
\end{bmatrix},
\end{equation*}
where $\mathbf{I}_2$ denotes the $2 \times 2$ identity matrix.

The CRLB for the parameters of interest $\boldsymbol{\theta}$ is obtained by computing the inverse of the Schur complement of $\mathbf{J}_{\boldsymbol{\kappa \kappa}}$ in $\mathbf{J}$. This yields the effective FIM:
\begin{equation}
\mathbf{J}_{\text{eff}} = \mathbf{J}_{\boldsymbol{\theta \theta}} - \mathbf{J}_{\boldsymbol{\theta \kappa}} \mathbf{J}_{\boldsymbol{\kappa \kappa}}^{-1} \mathbf{J}_{\boldsymbol{\kappa \theta}} = \frac{2|\alpha|^2}{N_0} \begin{bmatrix}
\Phi_{l} & \Xi \\
\Xi & \Phi_{\nu}
\end{bmatrix}.
\end{equation}
Here, the scalar coefficients $\Phi_{l}$, $\Phi_{\nu}$, and $\Xi$ represent the effective delay information, effective Doppler information, and the delay-Doppler coupling factor, respectively. Geometrically, $\Phi_{l}$ and $\Phi_{\nu}$ correspond to the squared norms of the derivative vectors $\bm{u}_l$ and $\bm{u}_\nu$ projected onto the orthogonal complement of the subspace spanned by $\bm{u}$. This projection effectively removes the uncertainty introduced by the unknown complex gain $\alpha$. They are defined, respectively, as
\begin{subequations}
\begin{align*}
\Phi_{l} & \triangleq \|\bm{u}_l\|^2 - \frac{|\langle \bm{u}_l, \bm{u} \rangle|^2}{\|\bm{u}\|^2}, \\
\Phi_{\nu} & \triangleq \|\bm{u}_\nu\|^2 - \frac{|\langle \bm{u}_\nu, \bm{u} \rangle|^2}{\|\bm{u}\|^2}, \\
\Xi & \triangleq \Re\{\langle \bm{u}_l, \bm{u}_\nu \rangle\} - \frac{\Re\{ \langle \bm{u}_l, \bm{u} \rangle \langle \bm{u}, \bm{u}_\nu \rangle \}}{\|\bm{u}\|^2}.
\end{align*}
\end{subequations}

\begin{prop}
\label{prop:CRLB}
The CRLBs for delay and Doppler estimation in an Agile-AFDM system, with unknown reflection coefficient, are given by
\begin{align*}
\operatorname{CRLB}(l) &= \frac{1}{2\mathrm{SNR}} \frac{\Phi_{\nu}}{\Phi_{l}\Phi_{\nu} - \Xi^2}, \\
\operatorname{CRLB}(\nu) &= \frac{1}{2\mathrm{SNR}} \frac{\Phi_{l}}{\Phi_{l}\Phi_{\nu} - \Xi^2},
\end{align*}
where $\mathrm{SNR}$ denotes the effective signal-to-noise ratio $|\alpha|^2/N_0$. 
\end{prop}

\begin{proof}
The CRLBs can be derived by inverting the $2 \times 2$ matrix $\mathbf{J}_{\text{eff}}$.  The denominator $\Phi_{l}\Phi_{\nu} - \Xi^2$ corresponds to the determinant of the normalized information matrix, while the numerator terms arise from the diagonal swap in the adjugate matrix.
\end{proof}

\begin{rem}
\label{rem:known_alpha}
Following Proposition \ref{prop:CRLB}, it is insightful to contextualize our result by comparing it with the theoretical limit achievable when the reflection coefficient $\alpha$ is perfectly known. This comparison clarifies the structure of the effective information terms $\Phi_l$ and $\Phi_{\nu}$.

In this case, the unknown parameter vector reduces to $\boldsymbol{\theta} = [l, \nu]^T$. The FIM for the known-$\alpha$ case, denoted as $\mathbf{J}_{\text{ideal}}$, is directly given by the top-left block of the full FIM but without the projection loss terms:
\begin{equation}
\mathbf{J}_{\text{ideal}} = 2\mathrm{SNR} \begin{bmatrix}
\|\bm{u}_l\|^2 & \Re\{\langle \bm{u}_l, \bm{u}_\nu \rangle\} \\
\Re\{\langle \bm{u}_l, \bm{u}_\nu \rangle\} & \|\bm{u}_\nu\|^2
\end{bmatrix}.
\end{equation}
Inverting this matrix yields the ideal CRLBs:
\begin{align*}
\operatorname{CRLB}_{\text{ideal}}(l) &= \frac{1}{2\mathrm{SNR}} \frac{\|\bm{u}_\nu\|^2}{\|\bm{u}_l\|^2 \|\bm{u}_\nu\|^2 - (\Re\{\langle \bm{u}_l, \bm{u}_\nu \rangle\})^2}, \\
\operatorname{CRLB}_{\text{ideal}}(\nu) &= \frac{1}{2\mathrm{SNR}} \frac{\|\bm{u}_l\|^2}{\|\bm{u}_l\|^2 \|\bm{u}_\nu\|^2 - (\Re\{\langle \bm{u}_l, \bm{u}_\nu \rangle\})^2}.
\end{align*}
Comparing these with Proposition \ref{prop:CRLB}, we observe that $\Phi_{l} \leq \|\bm{u}_l\|^2$ and $\Phi_{\nu} \leq \|\bm{u}_\nu\|^2$. The term $\frac{|\langle \bm{u}_l, \bm{u} \rangle|^2}{\|\bm{u}\|^2}$ in the definition of $\Phi_{l}$ (and similarly for $\Phi_{\nu}$) quantifies the information loss due to the estimation of the unknown reflection coefficient.
\end{rem}

\begin{rem}
By interchanging the expectation and minimization operators, we have
\begin{equation*}
\mathbb{E}_{\bm{x}} \left[ \min_{\bm{c} \in \mathcal{C}} \operatorname{CRLB}(\bm{c}, \bm{x}) \right] \leq \min_{\bm{c} \in \mathcal{C}} \mathbb{E}_{\bm{x}} \left[ \operatorname{CRLB}(\bm{c}, \bm{x}) \right].
\end{equation*}
This inequality is strict whenever the statically optimal parameter is suboptimal for a non-negligible set of data and channel realizations. Thus, Agile-AFDM guarantees a lower average estimation error bound than any static configuration.
\end{rem}

\subsection{CRLB Minimization}

\begin{algorithm}[t]
\caption{Agile-AFDM for CRLB optimization.}
\label{alg:pso_crlb}
\begin{algorithmic}[1]
\State \textbf{Input:} Number of particles $N_{\text{par}}$,
\State \hspace{0.95cm} maximum iterations $I_{\max}$,
\State \hspace{0.95cm} inertia weights $\omega_0$, $\omega_f$,
\State \hspace{0.95cm} acceleration coefficients $\phi_c$, $\phi_s$,
\State \hspace{0.95cm} convergence tolerance $\epsilon_{\text{con}}$
\State \textbf{Output:} Optimal chirp parameters $\mathbf{c}^* = [c_1^*, c_2^*]^T$,
\State \hspace{1.2cm} achieved CRLB value $\operatorname{CRLB}^*$

\State Initialize particles $\{\mathbf{c}_i^{(0)}\}_{i=1}^{N_{\text{par}}}$ randomly in $[0,1)^2$
\State Initialize velocities $\{\mathbf{v}_i^{(0)}\}_{i=1}^{N_{\text{par}}} = \mathbf{0}$
\State Set personal bests: $\mathbf{p}_i^{(0)} = \mathbf{c}_i^{(0)}$, $\mathcal{L}_{\text{p},i}^{(0)} = \operatorname{CRLB}(\mathbf{c}_i^{(0)})$
\State Set global best: $\mathbf{g}^{(0)} = \arg\min_i \mathcal{L}_{\text{p},i}^{(0)}$, $\operatorname{CRLB}^* = \min_i \mathcal{L}_{\text{p},i}^{(0)}$, $\mathrm{iter} = 0$

\Repeat \label{line:main_loop}
    \For{$i = 1$ to $N_{\text{par}}$ in parallel} \label{line:particle_loop}
        \State Evaluate objective: $\mathcal{L}_i \leftarrow \operatorname{CRLB}(\mathbf{c}_i^{(\mathrm{iter})})$
        \If{$\mathcal{L}_i < \mathcal{L}_{\text{p},i}^{(\mathrm{iter})}$}
            \State $\mathbf{p}_i^{(\mathrm{iter}+1)} \leftarrow \mathbf{c}_i^{(\mathrm{iter})}$, $\mathcal{L}_{\text{p},i}^{(\mathrm{iter}+1)} \leftarrow \mathcal{L}_i$ \Comment{Update personal best}
        \Else
            \State $\mathbf{p}_i^{(\mathrm{iter}+1)} \leftarrow \mathbf{p}_i^{(\mathrm{iter})}$, $\mathcal{L}_{\text{p},i}^{(\mathrm{iter}+1)} \leftarrow \mathcal{L}_{\text{p},i}^{(\mathrm{iter})}$
        \EndIf
    \EndFor \label{line:particle_end}
    
    \State $\mathbf{g}^{(\mathrm{iter}+1)} \leftarrow \arg\min_i \mathcal{L}_{\text{p},i}^{(\mathrm{iter}+1)}$ \Comment{Update global best}
    \State $\operatorname{CRLB}^* \leftarrow \min_i \mathcal{L}_{\text{p},i}^{(\mathrm{iter}+1)}$
    
    \For{$i = 1$ to $N_{\text{par}}$ in parallel} \label{line:update_loop}
        \State Compute inertia: $\omega(\mathrm{iter}) \leftarrow \omega_0 - (\omega_0 - \omega_f) \cdot \mathrm{iter}/I_{\max}$
        \State Generate random numbers: $r_1, r_2 \sim \mathcal{U}(0,1)$
        \State Update velocity: 
        \State \hspace{0.5cm} $\mathbf{v}_i^{(\mathrm{iter}+1)} \leftarrow \omega(\mathrm{iter}) \mathbf{v}_i^{(\mathrm{iter})} + \phi_c r_1 (\mathbf{p}_i^{(\mathrm{iter}+1)} - \mathbf{c}_i^{(\mathrm{iter})})$
        \State \hspace{2.0cm} $+ \phi_s r_2 (\mathbf{g}^{(\mathrm{iter}+1)} - \mathbf{c}_i^{(\mathrm{iter})})$
        \State Update position: $\mathbf{c}_i^{(\mathrm{iter}+1)} \leftarrow \mathbf{c}_i^{(\mathrm{iter})} + \mathbf{v}_i^{(\mathrm{iter}+1)}$
        \State Apply boundary constraints: 
        \State \hspace{0.5cm} $\mathbf{c}_i^{(\mathrm{iter}+1)} \leftarrow \max(0, \min(0.999, \mathbf{c}_i^{(\mathrm{iter}+1)}))$
    \EndFor \label{line:update_end}
    
    \State $\mathrm{iter} \leftarrow \mathrm{iter} + 1$
\Until{$\mathrm{iter} \geq I_{\max}$ or $\|\mathbf{g}^{(\mathrm{iter})} - \mathbf{g}^{(\mathrm{iter}-1)}\| < \epsilon_{\text{con}}$} \label{line:main_until}

\State $\mathbf{c}^* \leftarrow \mathbf{g}^{(\mathrm{iter})}$
\end{algorithmic}
\end{algorithm}

Having established the explicit dependence of the CRLB on the chirp parameters $(c_1, c_2)$, the subsequent challenge is to solve the corresponding optimization problem. For each transmission block, the optimal parameters for sensing are those that minimize the estimation uncertainty bound. This leads to the formulation of two optimization problems for delay and Doppler estimation, respectively:
\begin{equation} \label{eq:crlb_delay_optimization}
\{c_1^*, c_2^*\}_l = \underset{c_1, c_2 \in [0,1)}{\arg\min} \ \text{CRLB}(l),
\end{equation}
\begin{equation} \label{eq:crlb_doppler_optimization}
\{c_1^*, c_2^*\}_\nu = \underset{c_1, c_2 \in [0,1)}{\arg\min} \ \text{CRLB}(\nu).
\end{equation}
In practice, one may also target a composite objective, such as a weighted sum, for joint delay-Doppler performance.

Solving these problems is highly non-trivial due to the intricate, non-convex landscape of the CRLB function. The complexity arises from the trigonometric and exponential terms embedded within the effective channel matrix $\bm{H}_{\text{eff}}^{\text{sens}}$ and its derivatives, which define the information measures $\Phi_l$, $\Phi_\nu$, and $\Xi$. Traditional gradient-based methods prove ineffective due to the prevalence of local minima and the absence of reliable gradient information in regions of near-singularity. 

To overcome these challenges, we employ PSO \cite{marini2015particle}, a robust population-based metaheuristic inspired by social swarm behavior. PSO is particularly adept at navigating complex, non-convex search spaces without requiring gradient information. It effectively balances broad global exploration with intensive local exploitation, making it a fitting choice for our parameter optimization task where discovering a near-global optimum is more critical than exact gradient descent.

The PSO algorithm operates with a swarm of $N_{\text{par}}$ particles, each representing a candidate solution point $\textbf{c}=[c_1,c_2]^\top$ in the bounded search space $[0,1)^2$. At each iteration $\mathrm{iter}$, particle $i$ is characterized by a position vector $\mathbf{c}_i^{(\mathrm{iter})} = [c_{1,i}^{(\mathrm{iter})}, c_{2,i}^{(\mathrm{iter})}]^T \in [0,1)^2$, a velocity vector $\mathbf{v}_i^{(\mathrm{iter})} \in \mathbb{R}^2$, a personal best position $\mathbf{p}_i^{(\mathrm{iter})}$ with associated objective value $\mathcal{L}_{\text{p},i}^{(\mathrm{iter})} = \min_{0 \leq \tau \leq \mathrm{iter}} \operatorname{CRLB}(\mathbf{c}_i^{(\tau)})$, and knowledge of the global best position $\mathbf{g}^{(\mathrm{iter})} = \arg\min_{1 \leq i \leq N_{\text{par}}} \mathcal{L}_{\text{p},i}^{(\mathrm{iter})}$. The optimization objective is the CRLB, denoted as $\operatorname{CRLB}(\mathbf{c})$, for the estimation parameter of interest.

The swarm evolves through an iterative update process that simulates social learning. Each particle adjusts its trajectory based on its own experience (attraction to its personal best) and the collective experience of the swarm (attraction to the global best). The velocity and position updates follow
\begin{align*}
\mathbf{v}_i^{(\mathrm{iter}+1)} &= \omega(\mathrm{iter}) \mathbf{v}_i^{(\mathrm{iter})} + \phi_c r_1 \left(\mathbf{p}_i^{(\mathrm{iter}+1)} - \mathbf{c}_i^{(\mathrm{iter})}\right) \\
&\quad + \phi_s r_2 \left(\mathbf{g}^{(\mathrm{iter}+1)} - \mathbf{c}_i^{(\mathrm{iter})}\right), \\
\mathbf{c}_i^{(\mathrm{iter}+1)} &= \mathbf{c}_i^{(\mathrm{iter})} + \mathbf{v}_i^{(\mathrm{iter}+1)},
\end{align*}
where $\omega(\mathrm{iter})$ is the iteration-dependent inertia weight controlling the balance between exploration and exploitation, $\phi_c$ and $\phi_s$ are the cognitive and social acceleration coefficients, and $r_1, r_2 \sim \mathcal{U}(0,1)$ are independent random variables generated at each iteration. The inertia weight is linearly decreased across iterations according to
\begin{equation*}
\omega(\mathrm{iter}) = \omega_0 - (\omega_0 - \omega_f) \cdot \mathrm{iter} / I_{\max},
\end{equation*}
where $\omega_0$ and $\omega_f$ denote the initial and final inertia values, respectively, thereby transitioning the algorithm from global exploration toward local refinement as iterations progress.

The complete workflow for CRLB minimization using PSO is formalized in Algorithm~\ref{alg:pso_crlb}. The algorithm initializes particles uniformly within the parameter space and sets initial velocities to zero. Personal bests are initialized to the initial positions with their corresponding CRLB values. At each iteration, particles evaluate the objective function at their current positions, update personal bests when improvements are observed, update the global best across the swarm, and then move according to the PSO update equations with adaptive inertia. Boundary conditions are enforced through clamping to keep particles within $[0, 0.999]^2$. The algorithm terminates when either the maximum iteration count $I_{\max}$ is reached or the change in global best position falls below the convergence tolerance $\epsilon_{\text{con}}$.

\begin{rem}
The computational complexity of algorithm~\ref{alg:pso_crlb} is analyzed as follows. Each iteration requires $N_{\text{par}}$ objective function evaluations in parallel. Each CRLB evaluation involves constructing the AFDM channel matrix and computing Jacobian matrices, both with complexity $\mathcal{O}(N^2)$, while FIM inversion has negligible cost due to its $2 \times 2$ dimensionality. Averaging over $N_h$ channel realizations gives $\mathcal{O}(N_h N^2)$ per evaluation. The overall complexity is $\mathcal{O}(I_{\max} N_{\text{par}} N_h N^2)$. In contrast, exhaustive grid search with resolution $G$ per dimension requires $\mathcal{O}(G^2 N_h N^2)$ evaluations, which becomes prohibitive for fine-grained exploration. The PSO approach offers significant computational advantage through its adaptive sampling strategy, particularly when $N_{\text{par}} I_{\max} \ll G^2$.
\end{rem}

In summary, the integration of PSO with the derived CRLB analysis forms the core engine for sensing-optimal waveform adaptation in Agile-AFDM. This approach enables the dynamic, block-by-block identification of chirp parameters that minimize the theoretical lower bound on estimation error. The efficacy of this optimization, resulting in significant reductions in estimation uncertainty for both delay and Doppler, will be quantitatively demonstrated in Section \ref{Sec:Simu}.

%% file: Simu.tex
\begin{table}[t]
\caption{PAPR Simulation Parameters}
\label{tab:afdm_sim_params}
\centering
\setlength{\tabcolsep}{3.5mm}
\begin{tabular}{ccc}
\toprule
\textbf{Parameter} & \textbf{Description} & \textbf{Value} \\
\midrule
$N$ & Number of subcarriers & $64$ \\
$K$ & Number of users & $8$ \\
$T_{\mathrm{sym}}$ & Symbol duration & $128~\mu\mathrm{s}$ \\
$N_{\mathrm{cpp}}$ & CPP length & $10$ \\
\midrule
$L$ & Oversampling factor & $10$ \\
$\Delta c_2$ & Step size & $\frac{1}{80}$ \\
$\Delta c_2'$ & Step size & $\frac{1}{3120}$ \\
$\mathrm{CR}$ & Clipping ratio & $2$ \\
\bottomrule
\end{tabular}
\end{table}

To validate the theoretical analysis and demonstrate the efficacy of the proposed Agile-AFDM framework, comprehensive numerical simulations are conducted in this section. We evaluate the performance across the three core objectives introduced earlier: power efficiency (PAPR), communication reliability (SIR), and sensing accuracy (CRLB). The results are systematically compared against established benchmarks, including conventional OFDM and static-parameter AFDM.

\subsection{PAPR Performance}

\begin{figure*}[t]
  \centering
  \includegraphics[width=1.98\columnwidth]{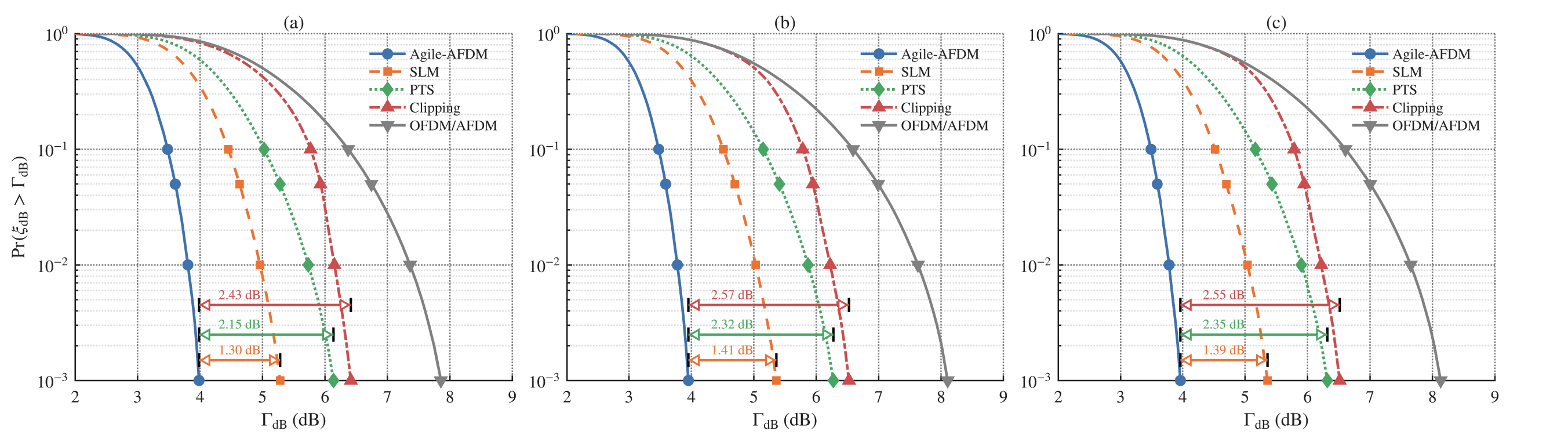}
  \caption{CCDF achieved by various PAPR reduction techniques under different modulation formats: (a) complex Gaussian signals, (b) 64QAM, (c) 128QAM. We highlight that AFDM with static chirp parameters exhibits the same PAPR performance with OFDM.}
  \label{fig:papr}
\end{figure*}

We first assess the PAPR reduction capability of Agile-AFDM. For comparison, we select three prevalent PAPR reduction techniques designed for OFDM systems: clipping, Selective Mapping (SLM), and Partial Transmit Sequence (PTS). These benchmarks were chosen for comparison because, like Agile-AFDM, they preserve spectral efficiency and require minimal signaling overhead.

The simulation parameters are summarized in Table~\ref{tab:afdm_sim_params}. Specifically, we consider an uplink multi-user scenario with $K=8$ users, where each user is allocated $8$ out of $N=64$ subcarriers. The number of subcarriers $N = 64$, the CPP length $N_{\mathrm{cpp}}=10$, and the symbol duration $T_{\mathrm{sym}} = 128~\mu\mathrm{s}$. To accurately capture PAPR behavior of analog waveforms, we consider an oversampling factor $L = 10$. The data symbols are transmitted as either complex Gaussian symbols or QAM modulated symbols to evaluate the versatility of Agile-AFDM across different symbol types.

The Complementary Cumulative Distribution Function (CCDF) is used as the metric for PAPR performance, reflecting the probability that the PAPR exceeds a given threshold \cite{sec2_papr}:
\begin{equation*}
\text{CCDF}(\Gamma) = \Pr(\xi > \Gamma).
\end{equation*}

Fig.~\ref{fig:papr} presents the CCDF performance for various PAPR reduction schemes, with data symbols configured as either complex Gaussian or QAM. Taking the complex Gaussian signal in Fig.~\ref{fig:papr}(a) as an example, we make the following key observations.
\begin{itemize}[leftmargin=0.5cm]
\item For the OFDM system without any PAPR mitigation scheme, the PAPR required to reach a CCDF of $10^{-3}$ is approximately $7.86~\mathrm{dB}$. In contrast, the proposed Agile-AFDM scheme with optimized frequency-domain chirp parameter $c_2$ significantly reduces the PAPR to $3.98~\mathrm{dB}$ under the same CCDF, effectively cutting it in half. This dramatic improvement stems from the dynamic modulation of the signal's time-frequency structure.
\item Agile-AFDM consistently outperforms state-of-the-art PAPR reduction techniques. Specifically, Agile-AFDM achieves PAPR reductions of approximately $1.3~\mathrm{dB}$ over SLM, $2.15~\mathrm{dB}$ over PTS, and $2.43~\mathrm{dB}$ over clipping.
\end{itemize}

\begin{rem}
Recall from the statistical invariance property of the unitary IDAFT that a fixed-parameter AFDM system exhibits PAPR statistics identical to OFDM. Therefore, the gains shown in Fig.~\ref{fig:papr} are solely attributable to the proposed data-aware, per-block optimization of $c_2$.
\end{rem}

\begin{rem}
For Agile-AFDM, the optimal chirp parameter $c_2$ is identified by searching the parameter space, with each candidate requiring a PAPR evaluation. Similarly, SLM and PTS achieve PAPR reduction by searching over phase sequences and evaluating PAPR for each candidate configuration. To ensure a fair comparison, we fix the total number of PAPR evaluations to $128$ for all methods.
\end{rem}

\subsection{SIR Performance}

\begin{figure*}[t]
  \centering
  \includegraphics[width=1.98\columnwidth]{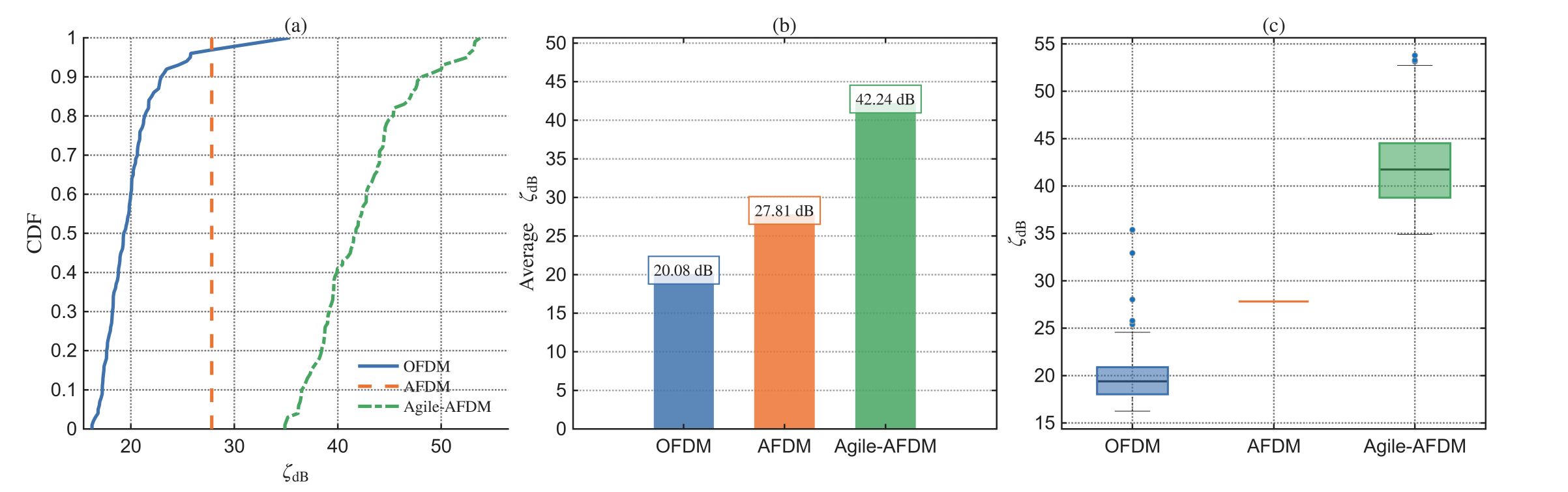}
  \caption{SIR Performance Comparison of OFDM, static (but parameter-optimized) AFDM and Agile-AFDM in Rayleigh Fading Channels: (a) CDF, (b) Average SIR Comparison, (c) Statistical Distribution Analysis.}
  \label{fig:sir}
\end{figure*}

\begin{table}[t]
\centering
\caption{SIR Simulation Parameters}
\label{tab:sim_params}
\setlength{\tabcolsep}{2mm}
\begin{tabular}{ccc}
\toprule
\textbf{Parameter} & \textbf{Description} & \textbf{Value} \\
\midrule
$N$ & Number of subcarriers & $64$ \\
$l$ & Path delays & $[1, 4, 5]$ \\
$\nu$ & Normalized Doppler & $[0.1, 0.4, 0.7]$ \\
-- & Tap powers & $[1.00, 0.20, 0.05]$ \\
$N_h$ & Channel realizations & $100$ \\
$N_b$ & Number of data blocks & $100$ \\
\midrule
$G$ & Grid search resolution & $100 \times 100$ \\
$B_1 \times B_2$ & Number of blocks & $16$ \\
$\epsilon_{\text{con}}$ & Convergence tolerance & $10^{-6}$ \\
$I_{\max}$ & Maximum iterations & $50$ \\
$\alpha$ & Adam learning rate & $0.001$ \\
$\beta_1$ & Adam $\beta_1$ & $0.9$ \\
$\beta_2$ & Adam $\beta_2$ & $0.999$ \\
$\epsilon_{\text{adam}}$ & Adam $\epsilon$ & $10^{-8}$ \\
\bottomrule
\end{tabular}
\end{table}

Next, we evaluate the communication reliability of Agile-AFDM by analyzing its SIR performance. The simulations validate the theoretical framework of Section \ref{Sec:SIR} under Rayleigh fading channels. The key simulation parameters are listed in Table~\ref{tab:sim_params}.
Fig.~\ref{fig:sir} presents a comprehensive performance comparison among  OFDM, static AFDM, and Agile-AFDM schemes $across$ 100 data blocks with varying channel realizations. 
\begin{itemize}[leftmargin=0.5cm]
    \item The Cumulative Distribution Functions (CDFs) presented in Fig.~\ref{fig:sir}(a) reveal the probabilistic performance profile of each scheme. OFDM exhibits the lowest SIR, concentrated between $16.24$-$35.27$ dB. Agile-AFDM shows a significant rightward shift, with SIR values predominantly in the $34.90$-$53.77$ dB range, indicating superior and more consistent performance. The vertical line for static AFDM represents its fixed-parameter nature, achieving a constant SIR of approximately $27.81$ dB.
    \item The average SIR comparison in Fig.~\ref{fig:sir}(b) quantifies the mean performance gains. Using OFDM ($20.08$ dB) as the baseline, static AFDM with optimized, but fixed parameters achieves a $7.73$ dB improvement ($27.81$ dB). In contrast, Agile-AFDM attains an average SIR of $42.24$ dB, delivering a substantial gain of $22.16$ dB over OFDM and $14.43$ dB over static AFDM.
    \item The boxplot in Fig.~\ref{fig:sir}(c) highlights the robustness of the proposed method. Agile-AFDM maintains a high SIR distribution with a median of $41.75$ dB. Notably, its entire interquartile range ($38.77$ dB -- $44.52$ dB) lies significantly above the fixed AFDM level ($27.81$ dB) and the OFDM distribution ($18.03$ dB -- $20.87$ dB). While OFDM shows several positive outliers, its core performance remains inferior. The results valid that Agile-AFDM effectively mitigates interference through adaptive tuning, consistently delivering high-quality signal reception.
\end{itemize}

These results collectively affirm that dynamic, data-aware parameter optimization in Agile-AFDM provides significant and robust gains in ICI suppression, translating directly to enhanced communication reliability in high-mobility scenarios.

\begin{table}[t]
\centering
\caption{CRLB Simulation Parameters}
\label{tab:crlb_sim_params}
\setlength{\tabcolsep}{2mm}
\begin{tabular}{ccc}
\toprule
\textbf{Parameter} & \textbf{Description} & \textbf{Value} \\
\midrule
$N$ & Number of subcarriers & $64$ \\
$\mathrm{SNR}$  & Effective signal-to-noise ratio & $20$ dB \\
$\sigma_x^2$ & Power of data symbols & $\frac{1}{N}$\\
\midrule
$G$ & Grid search resolution & $100 \times 100$ \\
$N_{\text{par}}$ & Number of particles & $200$ \\
$I_{\max}$ & Maximum iterations & $100$ \\
$\epsilon_{\text{con}}$ & Convergence tolerance & $10^{-6}$ \\
$\omega$ & Initial inertia weight & $0.99$ \\
$\phi_c$ & Cognitive acceleration coefficient & $1.2$ \\
$\phi_s$ & Social acceleration coefficient & $1.8$ \\
\bottomrule
\end{tabular}
\end{table}

\begin{figure*}[t]
  \centering
  \includegraphics[width=2\columnwidth]{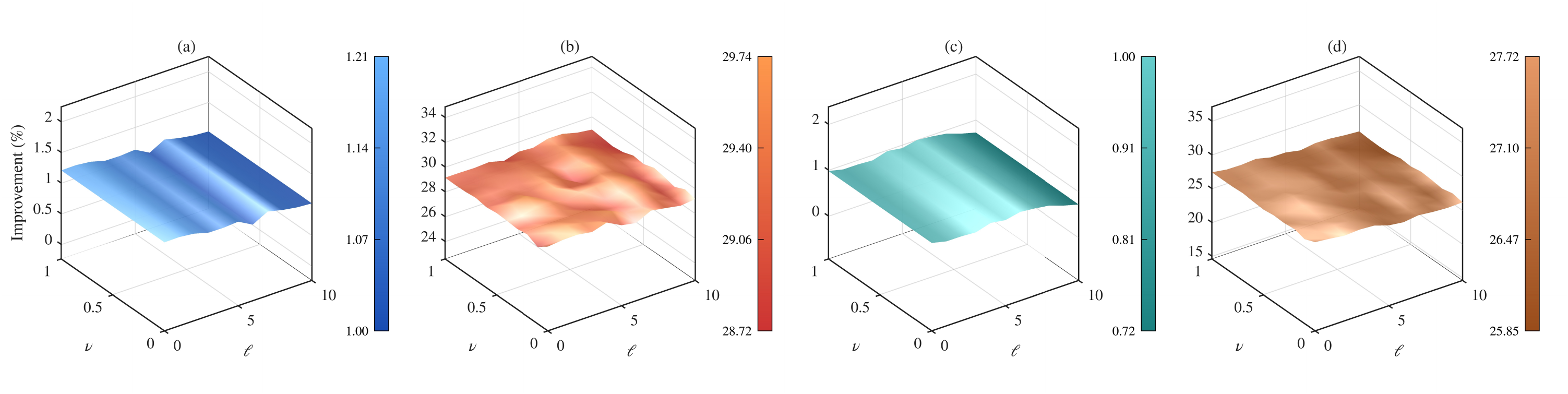}
  \caption{Percentage improvement of the Agile-AFDM in terms of CRLB reduction versus $\ell$ and $\nu$. 
    (a) Agile-AFDM vs. OFDM: CRLB on channel delay estimation. 
    (b) Agile-AFDM vs. OFDM: CRLB on Doppler shift estimation. 
    (c) Agile-AFDM vs. AFDM: CRLB on channel delay estimation. 
    (d) Agile-AFDM vs. AFDM: CRLB on Doppler shift estimation.}
  \label{fig:crlb}
\end{figure*}

\subsection{CRLB Performance}

\subsubsection{Overall performance comparison}
We first conduct a comprehensive evaluation of parameter estimation performance through systematic CRLB analysis.
For the baseline static AFDM system, where a closed-form solution for the CRLB is unavailable, we determine its optimal fixed parameters via an exhaustive grid search over a $100 \times 100$ discretization of the $(c_1, c_2)$ space. In contrast, the proposed Agile-AFDM framework employs the PSO algorithm (Algorithm~\ref{alg:pso_crlb}) to dynamically optimize these chirp parameters for each sensing block. The key simulation parameters, including the PSO configuration, are detailed in Table~\ref{tab:crlb_sim_params}.

The performance gains of Agile-AFDM are illustrated in Fig.~\ref{fig:crlb}, which plot the CRLB improvement compared to OFDM and AFDM across a range of delay ($l$) and normalized Doppler ($\nu$) values.
\begin{itemize}[leftmargin=0.5cm]
    \item \textit{Delay Estimation}: 
    As observed in Fig.~\ref{fig:crlb}(a) and (c), the proposed Agile-AFDM provides a marginal improvement in delay estimation, with average gains of approximately $1.11\%$ over OFDM and $0.91\%$ over static AFDM. This relatively small gap is expected, as the benchmark waveforms already possess strong delay resolution capabilities. 
    It is worth noting that although the performance gain exhibits a slight downward trend as the channel delay index $\ell$ increases, this variation is negligible. The proposed scheme demonstrates remarkable stability, maintaining consistent performance superiority even at relatively large delay spreads.
    \item \textit{Doppler Estimation}: 
    In contrast to the delay domain, Fig.~\ref{fig:crlb}(b) and (d) demonstrate that Agile-AFDM yields substantial performance gains in Doppler estimation. Specifically, the scheme achieves a CRLB reduction of roughly $29.24\%$ compared to OFDM and $26.74\%$ compared to static AFDM. 
    The improvement remains robust and uniform across the entire range of the normalized Doppler $\nu$. This indicates that the proposed agile parameter tuning effectively preserves the optimal signal energy concentration in the transform domain despite high mobility. By dynamically matching the waveform parameters to the channel statistics, the proposed scheme maintains a significant performance lead over static baselines regardless of the Doppler shift magnitude.
\end{itemize}

The contrast between the marginal gains in delay estimation and the substantial improvements in Doppler estimation raises a fundamental question regarding the underlying optimization landscape: why does parameter agility yield massive returns for Doppler estimation but only limited benefits for delay?

\subsubsection{Parameter sensitivity analysis}
To uncover the physical origin of this performance asymmetry, we investigate the sensitivity of the CRLB objective function with respect to the chirp parameters $(c_1, c_2)$. 
Our hypothesis is that the magnitude of the performance gain is intrinsically linked to the ``flatness'' of the optimization landscape. Specifically, if the CRLB is insensitive to parameter variations (a flat landscape), the static AFDM baseline is likely already operating near the optimum, leaving little room for PSO-based improvement. Conversely, a highly sensitive landscape implies that static parameters are prone to significant misalignment, which Agile-AFDM can correct.

\begin{table}[t]
\centering
\caption{Statistical Summary of Sensitivity Analysis at $(l=4, \nu=0.3)$}
\label{tab:crlb_sensitivity_stats}
\begin{tabular}{lcc}
\hline
\textbf{Metric} & CRLB($l$) & CRLB($\nu$) \\
\hline
\multirow{2}{*}{RV (\%)} & Mean: $10.79
$ & Mean: $2359343.75
$ \\
& Std. Dev.: $3.98
$ & Std. Dev.: $259702.88
$ \\
\hline
\multirow{2}{*}{CV (\%)} & Mean: $1.89
$ & Mean: $89.97
$ \\
& Std. Dev.: $0.50
$ & Std. Dev.: $0.02
$ \\
\hline
\end{tabular}
\end{table}

To quantify this, we perform a focused analysis at the channel condition $(l=4, \nu=0.3)$, where Agile-AFDM provides a $0.98\%$ improvement for delay estimation versus a $26.86\%$ improvement for Doppler estimation. We compute the CRLB over a fine grid of parameter pairs $(c_1, c_2)$ across $100$ Monte Carlo trials and employ two complementary metrics: the relative variation (RV) and the coefficient of variation (CV).
RV is defined as
\begin{equation*}
\text{RV} \triangleq \frac{\max_{c_1,c_2} \text{CRLB} - \min_{c_1,c_2} \text{CRLB}}{\min_{c_1,c_2} \text{CRLB}} \times 100\%,
\label{eq:rv}
\end{equation*}
which quantifies worst-case parameter sensitivity, while the CV is defined as
\begin{equation*}
\text{CV} \triangleq \frac{\mathrm{std}(\text{CRLB})}{\mathrm{mean}(\text{CRLB})} \times 100\%,
\label{eq:cv}
\end{equation*}
where $\mathrm{mean}(\cdot)$ and $\mathrm{std}(\cdot)$ denote the average and standard deviation operators across the parameter space, respectively.

The statistical results are summarized in Table~\ref{tab:crlb_sensitivity_stats}. The CRLB($l$) exhibits low sensitivity with mean RV of $10.79\%$ (std. dev. $3.98\%$) and mean CV of $1.89\%$ (std. dev. $0.50\%$). In contrast, Doppler estimation demonstrates substantially higher sensitivity, with mean RV of $2359343.75\%$ (std. dev. $259702.88\%$) and mean CV of $89.97\%$ (std. dev. $0.02\%$). This striking difference reveals that the CRLB($\nu$) landscape exhibits pronounced variations around the optimal parameters, while the CRLB($l$) remains relatively stable across the $(c_1, c_2)$ space. Consequently, PSO-based optimization is essential for Doppler estimation to navigate the highly sensitive landscape, while delay estimation performance remains inherently robust to parameter variations.

%% file: Conclusions.tex
This paper has introduced Agile-AFDM, a dynamic and data-aware waveform adaptation framework that shifts the design paradigm of AFDM from static, worst-case configuration to real-time, context-aware optimization. By treating chirp parameters as dynamic degrees of freedom, we enable AFDM to self-adapt for diverse objectives such as power efficiency, communication reliability, and sensing accuracy, marking a concrete step forward in waveform intelligence for next-generation wireless systems.

A natural question arises: does such agility come at an unsustainable cost in complexity? We argue that the penalty is not only manageable but also worthwhile. While real-time optimization introduces additional computational load, the linear and fully digital nature of AFDM transformations ensures implementability in modern baseband processors. Moreover, the periodicity of chirp parameters significantly confines the search space, and practical accelerations (through parallel processing, hardware-friendly approximations, or even lightweight neural predictors) can make dynamic adaptation feasible even in latency-sensitive scenarios. In essence, agility does not necessitate impracticality.

Looking ahead, several exciting directions emerge from this work. 
\begin{itemize}[leftmargin=0.5cm]
    \item First, while we have focused on optimizing PAPR, ICI, and CRLB separately, a natural extension lies in multi-objective optimization, where system designers can tradeoff between efficiency, reliability, and sensing precision according to real-time demands. Future work could explore Pareto-optimal parameter fronts, adaptive weighted cost functions, or reinforcement-learning agents that dynamically balance competing objectives under varying channel and traffic conditions.
    \item Second, the core idea of agility extends beyond AFDM. Similar parameter-aware adaptation can be applied to fractional Fourier division multiplexing, generalized frequency division multiplexing, or any modulation scheme that embeds configurable degrees of freedom for shaping time-frequency-spatial structures. This points toward a broader vision of context-aware physical layers that self-configure based on the operational environment.
\end{itemize}

%% file: Appendix/App_prop_perio.tex
Since the communication channel $\bm{H}_{\text{eff}}^{\text{comm}}$ is a linear combination of multipath components and the sensing channel $\bm{H}_{\text{eff}}^{\text{sens}}$ represents a single return path, it suffices to prove the periodicity for the contribution of a generic path with complex gain $\beta$. The $(p,q)$-th element of such a generic channel matrix is proportional to:
\begin{equation*}
\mathcal{E}(c_1, c_2) = e^{j \frac{2\pi}{N} \left(N c_1 \ell^2 - q \ell + N c_2 (q^2 - p^2)\right)} \mathcal{F}(p,q; c_1),
\end{equation*}
where $\mathcal{F}(p,q; c_1)$ denotes the dependence of the function $\mathcal{F}$ on $c_1$ via $\psi$. We analyze the impact of integer shifts $k, m \in \mathbb{Z}$ on the parameters $c_1$ and $c_2$.

First, consider the phase term. Substituting $c_1 \leftarrow c_1 + k$ and $c_2 \leftarrow c_2 + m$, the exponential term becomes:
\begin{equation*}
\begin{aligned}
&e^{j \frac{2\pi}{N} \left(N (c_1+k) \ell^2 - q \ell + N (c_2+m) (q^2 - p^2)\right)} \\
&= e^{j \frac{2\pi}{N} (\dots)} \cdot \underbrace{e^{j 2\pi k \ell^2}}_{=1} \cdot \underbrace{e^{j 2\pi m (q^2 - p^2)}}_{=1},
\end{aligned}
\end{equation*}
where the identity holds because $\ell, p, q$ are integers, making the exponents integer multiples of $2\pi j$.

Second, consider the function $\mathcal{F}(p,q)$, which depends on $\psi = p - q + \nu + 2 N c_1 \ell$. Under the shift $c_1 \leftarrow c_1 + k$, the variable $\psi$ transforms as:
\begin{equation*}
\psi' = p - q + \nu + 2 N (c_1+k) \ell = \psi + 2 N k \ell.
\end{equation*}
The summation term in \eqref{eq:commF} is $e^{-j \frac{2\pi}{N} \psi n}$. Substituting $\psi'$, we obtain:
\begin{equation*}
e^{-j \frac{2\pi}{N} \psi' n} = e^{-j \frac{2\pi}{N} \psi n} \cdot \underbrace{e^{-j \frac{2\pi}{N} (2 N k \ell) n}}_{e^{-j 4\pi k \ell n} = 1}.
\end{equation*}
Since the phase rotation per sample remains invariant, $\mathcal{F}(p,q)$ is unchanged.

Consequently, every element of the effective channel matrices remains invariant under integer shifts of $c_1$ and $c_2$. This concludes the proof.

%% file: Appendix/App_thm_tri.tex
Using the product-to-sum formula, we can express the product of $\omega_1\left(kt\right)$, $\omega_2\left(lt\right)$, $\omega_3\left(mt\right)$, and $\omega_4\left(nt\right)$ as a sum (or difference) of eight trigonometric functions. 
As an example,
\begin{equation*}
\begin{aligned}
&\cos(2t)\cdot\cos(4t)\cdot\cos(1t)\cdot\cos(5t)\\
&=\left(\frac{1}{2}[\cos(6t)+\cos(-2t)]\right)\cdot\left(\frac{1}{2}[\cos(6t)+\cos(-4t)]\right)\\
&=\frac{1}{4}[\cos(6t)\cos(6t)+\cos(6t)\cos(-4t)\\
&\hspace{0.8cm} +\cos(-2t)\cos(6t)+\cos(-2t)\cos(-4t)]\\
&=\frac{1}{8}[\cos(12t)+1+\cos(2t)+\cos(10t)\\
&\hspace{0.8cm}+\cos(4t)+\cos(-8t)+\cos(-6t)+\cos(2t)].
\end{aligned}
\end{equation*}

In general, we represent this process as:
\begin{equation*}
\begin{aligned}
& \omega_1(kt)\cdot\omega_2(lt)\cdot\omega_3(mt)\cdot\omega_4(nt)\\
& =\frac14[u_1((k+l)t)+u_{2}((k-l)t)]\\
&  \hspace{0.5cm} \times [u_{3}((m+n)t)+u_{4}((m-n)t)]\\
& =\frac18[u_5((k+l+m+n)t)+u_6((k+l-m-n)t)\\
&  \hspace{0.5cm} +u_7((k+l+m-n)t)+u_8((k+l-m+n)t)\\
&  \hspace{0.5cm} +u_9((k-l+m+n)t)+u_{10}((k-l-m-n)t)\\
&  \hspace{0.5cm} +u_{11}((k-l+m-n)t)+u_{12}((k-l-m+n)t)],
\end{aligned}
\end{equation*}
where $u_{1\sim 12}(t)$ represents either $\pm \cos(t)$ or $\pm \sin(t)$. 

In particular,
\begin{equation*}
   \int_0^{2\pi} u_{i}(wt) dt = 
   \begin{cases}
       0, & \text{If $u_{i}(wt)=\pm \sin(wt)$}, \\
       \pm 2\pi \delta(w), & \text{If $u_{i}(wt)=\pm \cos(wt)$}.
   \end{cases}
\end{equation*}

This outcome guides the computation using the matrix $\bm{Q}$. Specifically,
\begin{itemize}
    \item If there are either one or three $\cos(t) $ terms among $\omega_{1\sim4}(t)$, then  $u_{5\sim12}(t)$  will be   $\pm \sin(t) $, resulting in a zero integral. This corresponds to the zero rows in $\bm{Q}$.
    \item Otherwise, $ u_{5\sim12}(t)$ are $\pm \cos(t)$, and the integral's outcome depends on the frequency combinations, represented by the $\pm{1}$ entries in $\bm{Q}$.
    \item The first column of $\bm{Q}$ is entirely zero because $k + l + m + n$ is non-zero for $1 \leq k, l, m, n \leq N-1$.
\end{itemize}

The row index $i$ is determined by the indicator functions  $\delta(\omega_j(t) = \sin(t))$, which check if each  $\omega_j(t)$ is  $\sin(t)$ or $\cos(t)$.